 \documentclass[11pt]{article}
 
\usepackage{amssymb,amsmath,amsfonts}
\usepackage{graphicx,color,enumitem}
\usepackage{mathrsfs}
\usepackage{amsthm,amsmath} 
\usepackage[dvipsnames]{xcolor}
\usepackage{bm}
\usepackage[round]{natbib}
\usepackage[]{appendix}
\usepackage{csquotes}

\usepackage{geometry}

\definecolor{mycolor}{rgb}
{0.796, 0.255, 0.329}

\RequirePackage[colorlinks,citecolor=blue,urlcolor=blue, linkcolor=blue]{hyperref}

\usepackage{tikz}
\tikzstyle{vertex}=[circle, draw, inner sep=2pt, fill=white]

\newcommand{\e}{{\varepsilon}}

\newcommand{\E}{{\mathbb E}}

\newcommand{\C}{{\mathbb C}}
\newcommand{\R}{{\mathbb R}}

\newcommand{\N}{{\mathbb N}}

\newcommand{\Y}{{\mathbb Y}}
\newcommand{\Z}{{\mathbb Z}}

\newcommand{\Gcal}{{\mathcal G}}

\newcommand{\Ical}{{\mathcal I}}
\newcommand{\Jcal}{{\mathcal J}}

\newcommand{\Lcal}{{\mathcal L}}

\newcommand{\ostar}{{\,\overline \star\,}}

\newcommand{\red}{\color{green!70!black}}

\newcommand{\W}{\mathbb W}

\newcommand{\fdot}{{\,\cdot\,}}

\newtheorem{theorem}{Theorem}

\newtheorem{condition}{Condition}

\newtheorem{corollary}[theorem]{Corollary}

\theoremstyle{definition}
\newtheorem{definition}[theorem]{Definition}
\newtheorem{remark}[theorem]{Remark}
\newtheorem{example}[theorem]{Example}

\newtheorem{lemma}[theorem]{Lemma}

\numberwithin{equation}{section}
\numberwithin{theorem}{section}

\begin{document}

\title{\textbf{Local signature-based expansions}\thanks{We are grateful to Christa Cuchiero (discussant) and, for their comments, to conference participants in the XXV Workshop on Quantitative Finance (Bologna, April 11-13 2024), the Barcelona Workshop in Financial Econometrics (Barcelona, May 30-31 2024), the $6^{th}$ \enquote{Quantitative Finance and Financial Econometrics} Conference (Marseille, June 4-7 2024), the $4^{th}$ \enquote{High Voltage Econometrics} Conference (Lecce, October 4-5 2024) and the \enquote{Financial Econometrics meets Machine Learning} Conference (Lugano, November 1-2 2024).}}
\author{Federico M. Bandi{\thanks{\noindent Johns Hopkins University, Carey Business School, 555 Pennsylvania Avenue, Washington, DC 20001, USA;  e-mail:\   fbandi1@jhu.edu.}} \and Roberto Ren\`o{\thanks{\noindent Essec Business School, 3 Av. Bernard Hirsch, 95000 Cergy, France;  e-mail:\ reno@essec.edu. ~}}\and Sara Svaluto-Ferro{\thanks{\noindent Department of Economics, University of Verona, Via Cantarane 24, 37129 Verona, Italy;  e-mail:\ sara.svalutoferro@univr.it. ~}}
}
\date{\today}

\maketitle

\begin{abstract}
We study the local (in time) expansion of a continuous-time process and its conditional moments, including the process' characteristic function. The expansions are conducted by using the properties of the (time-extended) \textit{It\^o signature}, a tractable basis composed of iterated integrals of the driving deterministic and stochastic signals: time, multiple correlated Brownian motions and multiple correlated compound Poisson processes. We show that these properties are conducive to \textit{automated} expansions to \textit{any} order with \textit{explicit} coefficients and, therefore, to stochastic representations in which asymptotics can be conducted for a shrinking time ($t \rightarrow 0$), as in the extant continuous-time econometrics literature, but, also, for a fixed time (such that $t<1$) with a diverging expansion order. The latter design opens up novel opportunities for identifying deep characteristics of the assumed process. 

\vspace{0.5cm}

\noindent\textbf{Keywords:} Brownian motion, compound Poisson process, It\^o signature. \\
\noindent \textbf{JEL Classifications:} C18, C22.

\end{abstract}

\pagebreak


\section{Introduction}

We discuss the local (in time) expansion of a stochastic process and its conditional moments. We study a \textit{triplet} of expansions. First, we expand the process itself. Second, we expand \enquote{regular} moments of the process, i.e., conditional expected functions of the process with well-behaved derivatives at time 0. Finally, we turn to \enquote{irregular} moments, i.e., conditional expected functions with unbounded derivatives at time 0. A key example of the latter is the characteristic function of the \textit{standardized} process.

We do so by exploiting the properties of the (time-extended) \textit{It\^o signature} of the process' driving shocks, a tractable basis composed of iterated integrals of the driving deterministic and random signals (namely, time, correlated Brownian motions and correlated Poisson processes) whose inherent algebraic structure will be documented to lead to automatic expansions to any order.

Assume, for simplicity, an It\^o semimartingale $(X_t)_{t\in[0,T]}$ driven by a single Brownian motion $(W_t)_{t\in [0,T]},$ but general specifications with multiple Brownian shocks and compound Poisson shocks will be allowed in what follows. A process $(X_t)_{t\in[0,T]}$ is $(W,t)$-differentiable if it admits the representation
\begin{equation}
X_t=X_0+\int_0^t c_0(s) ds +\int_0^t c_1(s) dW_s,
\end{equation}
for square integrable processes $(c_0(t))_{t\in[0,T]}$ and $(c_1(t))_{t\in[0,T]},$ the first-layer process' \enquote{characteristics}. Should these two (stochastic) characteristics feature the same representation as $(X_t)_{t\in[0,T]}$ with square-integrable processes $c_{00}(t), c_{10}(t)$ and $c_{01}(t), c_{11}(t)$ (the second-layer process' \enquote{characteristics}), respectively, then $(X_t)_{t\in [0,T]}$ would be viewed as being twice $(W,t)$-differentiable, and so on for higher orders.

Given $n+1$ $(W,t)$-differentiability, by iterating stochastic integrals, we show that $(X_t)_{t\in[0,T]}$ can, equivalently, be written as
\begin{equation}\label{first1}
X_t = \langle c, \widehat \W_t\rangle + \e_n(t),
\end{equation} 
where the first term on the right-hand side of Eq.~(\ref{first1}) is the inner product between a vector $c$ containing time-0 values of all characteristics (up to a generic depth $n$), i.e.,
\begin{equation}
c = (S_0, c_0(0), c_1(0), c_{00}(0),  c_{10}(0), c_{01}(0),  c_{11}(0), ...)
\end{equation}
and a vector of progressively increasing (up to depth $n$) iterated integrals
\begin{equation}
\widehat \W_t = \left(1, \int_0^t ds, \int_0^t dW_s, \int_0^t s ds,   \int_0^t W_s ds, \int_0^t s dW_s, \int_0^t W_s dW_s, ... \right),
\end{equation}
i.e., the truncated (to the order $n$) time-extended \textit{It\^o signature} of $W$. The second term on the right-hand side of Eq.~(\ref{first1}) is an approximation error whose statistical order will be made explicit.

The elements of the time-extended \textit{It\^o signature} satisfy a key property: their products can always be expressed as linear maps of other elements of the signature. We formalize this property and use it extensively in our proofs. We employ it, in particular, to transition from an expansion of the process (as in Eq.~(\ref{first1})) to expansions of the local moments of the process, both for the case in which nonlinear transformations of the process have bounded derivatives at zero (the \enquote{regular} case) and for the case when these transformations have unbounded derivatives at zero (the \enquote{irregular} case). As we emphasize below, the irregular case - which includes the characteristic function of the \textit{standardized} process - addresses the recent continuous-time econometrics literature directly.  

The notion of \textit{It\^o signature} traces back to \cite{C:57,C:77} and plays an important role in the context of rough path theory, initiated by \cite{L:98}. The usefulness of the concept is easily justified by existing universal approximation theorems. It is known that continuous (relative to specific variation distances) functionals of continuous as well as c\`adl\`ag paths can be approximated on compact sets of paths by linear functionals of the time-extended \textit{It\^o signature} (c.f. \citealp{CPS:22}).

We contribute to three largely separate (to date) lines of work: the growing recent literature in continuous-time econometrics on the use of short-term expansions for inference and pricing, the literature on signature-based methods in mathematical finance and the literature on short-term expansions in probability and process theory.

Regarding the first literature, local expansions of the process' characteristic function have been used to identify the characteristics (e.g., the volatility, $c_{1}(0),$ and the volatility of volatility, $c_{11}(0)$) of a process of interest on its first two layers (c.f., \citealp{jacod2014efficient}, \citealp{BR:17}, \citealp{todorov2019nonparametric}, \citealp{To:21}, and \citealp{chong2024volatility}) as well as to price structured financial products with short expirations (\citealp{BFR}). Expanding locally conditional moments, like the process' characteristic function, is known to be a cumbersome procedure, even when the expansion is low (first or second) order. As a result, existing expansions in the literature are low order. Low-order expansions, in turn, are empirically informative - for specific problems - but may be insufficient more generally, i.e., when time is not overly short and/or when econometric interest is in deeper characteristics that low-order expansion would not reveal. In our case, the expansions are automated and may be of arbitrary accuracy ($t^n$ with $t<1$ and any value of $n$), thereby opening up new possibilities, from inference on the deep characteristics of the c\`adl\`ag process of interest (as $n \rightarrow \infty$) to essentially exact (closed-form) solutions of short-term pricing problems. 

We turn to the second literature. Signature-based methods have experienced recent attention in mathematical finance (e.g., \citealp{PSS:20}) with applications largely focused on derivative pricing (e.g., \citealp{LNP:20}, \citealp{BHRS:21}, \citealp{CGS:23}, \citealp{AG:24}, and \citealp{CGMS:23}) and trading/portfolio allocation (e.g. \citealp{KLP:20}, \citealp{AGTZ:23}, \citealp{CM:22}, and \citealp{NCL:23}). Some studies have also exploited the potential of signature-based methods for financial machine learning (e.g., \citealp{buehler2020generating}, and \citealp{LBW:23}). We introduce the notion of \textit{It\^o signature} to the continuous-time econometrics literature, a literature that has seen considerable growth over the last two decades thanks to the increased availability of both (high-frequency) data and computational power. Specifically, we exploit the approximation properties of the \textit{It\^o signature} to derive expansions of flexible continuous-time models with c\`adl\`ag paths and their moments. Consistent with classical econometric modeling of high-frequency data, we expand over a short-time horizon, an observation which leads to our positioning within the third literature. 

Local - or short-term - expansions of moments and related quantities have been the subject of robust investigations. As \citet{bentata2012short} emphasize, their computation is central to \enquote{stochastic control problems, statistics of processes and mathematical finance.} The literature includes, among others, \citet{bismut1981mechanique}, \citet{azencott2006formule}, \citet{platen1982generalized}, \citet{leandre1987densite}, \citet{arous1989flots}, \citet{ishikawa1994asymptotic}, \citet{picard1996existence}, \citet{picard1997density},   \citet{ishikawa2001density}, \citet{ruschendorf2002expansion}, \citet{lyons2004cubature}, \citet{jacod2007asymptotic},\\ \citet{barndorff2008probability}, \citet{friz2008euler}, \citet{figueroa2009small}, \citet{marchal2009small}, \citet{figueroa2014small} and \citet{figueroa2018small}. While, in some cases (e.g., \citealp{arous1989flots}), this line of work uses - like we do - suitable properties of the algebra of iterated integrals, ($i$) we work with flexible continuous or discontinuous It\^o semimartingales rather than with continuous stochastic differential equations (SDEs) - the typical data generating process in this revealing, early work - and ($ii$) we derive linear expansions (for regular and irregular) moments with \textit{explicit} and \textit{interpretable} coefficients. Our interest is in providing methods which facilitate zooming into deep layers of the process of interest in order to extract interpretable information about its dynamics. In essence, relative to this literature, our different data generating process and objectives result in different methods of proof and representations. 
      
We proceed as follows. We begin with the continuous case, in the context of which we detail notation and methods. Section~\ref{process}, Section~\ref{regular} and Section~\ref{irregular} are devoted to the expansion of a general Brownian semimartingale and its \enquote{regular} and \enquote{irregular} moments, respectively. As an example of the proposed methods, we provide - in Section~\ref{any} - a new \textit{third-order} expansion of the characteristic function of the standardized process, a result which complements the second-order expansions in recent work and is illustrative of the automation yielded by the proposed approach. We then add compound Poisson shocks, something which necessitates some burdening of the notation on which we provide clarity. In the discontinuous case, we extend all results presented in previous sections (c.f. Section~\ref{Levy}). Section~\ref{Conclusions} offers concluding remarks. All proofs are in Appendix \ref{Appendix1} and \ref{Appendix2}.

\color{black}

\section{Expanding the process} \label{process}
In this section, we will be working with an It\^o process $(X_t)_{t\in[0,T]}$ driven by a $d$-dimensional Brownian motion $W:=(W^1,\ldots,W^d)$. Before introducing the necessary concepts and providing formalizations, we discuss the logic of our most basic expansion - that of the process itself - in a simple example with one Brownian motion $W$.

\subsection{The logic of the process expansion}\label{logic}

A ($W,t$)-differentiable process is an It\^o process $(X_t)_{t\in[0,T]}$ given by
\begin{equation}
X_t=X_0+\int_0^t c_0(s) ds +\int_0^t c_1(s) dW_s, \label{first}
\end{equation}
for some Brownian motion $(W_t)_{t\in[0,T]}$ and stochastic (drift and diffusion) processes $(c_0(t))_{t\in[0,T]}$ and $(c_1(t))_{t\in[0,T]}$. Suppose that $(c_0(t))_{t\in[0,T]}$ and $(c_1(t))_{t\in[0,T]}$ are also ($W,t$)-differentiable and, therefore, represented by
\begin{align}
c_0(s)&=c_0(0)+\int_0^s c_{00}(r) dr +\int_0^s c_{10}(r) dW_r, \label{second} \\
c_1(s)&=c_1(0)+\int_0^s c_{01}(r) dr +\int_0^s c_{11}(r) dW_r, \label{third}
\end{align}
given one extra layer of stochastic processes $(c_{ij}(t))_{t\in[0,T]}$. The original process $(X_t)_{t\in[0,T]}$ is, now, twice ($W,t$)-differentiable. Plugging Eq.~(\ref{second}) and Eq.~(\ref{third}) into Eq.~(\ref{first}), we obtain
\begin{align*}
X_t&=X_0+c_0(0)\int_0^t1ds+c_1(0)\int_0^t1dW_s\\
&+\int_0^t\int_0^s c_{00}(r) drds +\int_0^t\int_0^s c_{10}(r) dW_rds +\int_0^t\int_0^s c_{01}(r) dr dW_s +\int_0^t\int_0^s c_{11}(r) dW_rdW_s.
\end{align*}
Folding the last four double integrals into an error term $\e_1(t)$, we therefore have
\begin{equation*}
X_t=X_0+c_0(0)\int_0^t1ds+c_1(0)\int_0^t1dW_s + \e_1(t),
\end{equation*}
where $X_0+c_0(0)\int_0^t1ds+c_1(0)\int_0^t1dW_s$ is a linear combination of terms
\begin{equation}\label{eqn1}
1\qquad \int_0^t 1ds\qquad \int_0^t 1 dW_s,
\end{equation}
with weights given by the values (at time 0) of the process $(X_t)_{t\in[0,T]}$ and its drift and diffusion ($c_0$ and $c_1$). We will later show that, under suitable assumptions, the error term $\e_1(t)$ converges to 0 in probability faster than $t^{1/2},$ as $t\rightarrow 0,$ in the sense that $$\E_0[|\e_1(t)|]=o(t^{1/2}).$$
This discussion implies that the study of $(X_t)_{t\in[0,T]}$ for small $t$ can be performed by using the properties of the vector in Eq.~(\ref{eqn1}). In what follows, such vector will be characterized as containing elements of the time-extended \textit{It\^o signature} of $W$.

Using this same idea, assuming that the process $(X_t)_{t\in[0,T]}$ is ($W,t$)-differentiable three times, which is equivalent to assuming that
$$c_{ij}(s)=c_{ij}(0)+\int_0^s c_{0ij}(r) dr +\int_0^s c_{1ij}(r) dW_r,$$
we may write
\begin{eqnarray*}
X_t&=&X_0+c_0(0)\int_0^t1ds+c_1(0)\int_0^t1dW_s\\
&+&c_{00}(0) \int_0^t\int_0^s 1drds +c_{10}(0) \int_0^t\int_0^s 1dW_rds +c_{01}(0) \int_0^t\int_0^s 1 dr dW_s +c_{11}(0) \int_0^t\int_0^s 1 dW_rdW_s\\
&+&\e_2(t).
\end{eqnarray*}
The process $(X_t)_{t\in[0,T]}$ is now a linear combination of a longer vector of iterated integrals with respect to time and Brownian motion, i.e., 
\begin{equation}\label{eqn2}
1,
\int_0^t 1 ds,
\int_0^t 1 dW_s,
 \int_0^t\int_0^s 1drds, \int_0^t\int_0^s 1dW_rds, \int_0^t\int_0^s 1 dr dW_s, \int_0^t\int_0^s 1 dW_rdW_s,
\end{equation}
whose coefficients are, once more, the time-zero values of the processes at all three layers of the expansion in addition to an error term $\e_2(t)$ (a linear combination of eight triple integrals). Under conditions laid out below, the error term satisfies
 $$\E_0[|\e_2(t)|]=o(t).$$
Eq.~(\ref{eqn2}) is, again, a vector of elements of the time-extended \textit{It\^o signature} of $W$. 

It is easily seen that this same procedure may be repeated given suitable differentiability properties of the original process. To this extent, assume the process $(X_t)_{t\in[0,T]}$ is ($W,t$)-differentiable $n+1$ times. Define, for convenience, $\widehat W_t:=(t,W_t),$ so that $\widehat W_t^0=t$ and $\widehat W_t^1=W_t.$ We obtain
 \begin{eqnarray*}
 X_t&=&X_0+\sum_{k=1}^{n}\sum_{(i_1,\ldots,i_k)\in\{0,1\}^k}c_{i_1,\ldots,i_k}(0)\int_0^t\int_0^{t_{k}}\cdots \int_0^{t_{2}}1d\widehat W_{t_1}^{i_1}\cdots d\widehat W_{t_k}^{i_k} +\e_n(t),
 \end{eqnarray*}
where
$$\E_0[|\e_n(t)|]=o(t^{n/2})$$
and 
$$\int_0^t\int_0^{t_{k}}\cdots \int_0^{t_{2}}1d\widehat W_{t_1}^{i_1}\cdots d\widehat W_{t_k}^{i_k}$$
is a generic element of the time-extended \textit{It\^o signature} of $W$, a concept to which we now turn formally. 

Before doing so, we emphasize that the assumed process is rather flexible. It is more flexible than the continuous SDEs for which early stochastic Taylor expansions are established in the fundamental work of, e.g., \citet{azencott2006formule} and \citet{arous1989flots}. It is also more flexible than ubiquitous processes in the finance literature in which the characteristics are functions of unobservable state variables, like spot variance (as, e.g., in the affine tradition reviewed by \citealp{duffie2003affine}, \textit{inter alia}). In Section \ref{Levy}, we enhance flexibility further by allowing for discontinuities.

Because we do not specify a parametric structure for $c_0(t)$ and $c_1(t)$ (or any other generic element $c_{i_1,\ldots,i_k}(t)$), the process may be viewed as \textit{nonparametric}. As is the case for deterministic functions, the differentiability of every generic element $c_{i_1,\ldots,i_k}(t)$ in each layer, not surprisingly, matters. The more differentiable the overall process, the deeper one can zoom into its layers and, of course, the more higher-order expansions have the potential to provide granular information about the process dynamics through the dynamics of deep characteristics. Importantly, this information is offered by time-0 (stochastic) coefficients ($c_{i_1,\ldots,i_k}(0)$) which are, differently from those in the early stochastic Taylor expansions, \textit{interpretable}. As an example, $c_{1,1,1}(0)$ defines the volatility of the volatility of volatility, a \textit{third}-layer process which will be shown to enter the \textit{second}-order expansion of the characteristic function of the standardized process (c.f. Theorem \ref{prop2} and Corollary \ref{BR}).
\color{black}
\subsection{It\^o signature}\label{sig}
Let $(Y_t)_{t\in[0,T]}$ be a generic $d$-dimensional semimartingale whose elements are given by $Y^1,\ldots,Y^d$.

\begin{definition}
The process $(\Y_t)_{t\in[0,T]}$ expressed as
$$\Y_t:=\bigg(1,\int_0^t 1 dY^1_{t_1},\ldots,\int_0^t 1 dY^d_{t_1}, \int_0^t \int_0^{t_2}1 dY^1_{t_1}dY^1_{t_2},\int_0^t \int_0^{t_2}1 dY^1_{t_1}dY^2_{t_2},\ldots\bigg)$$
is called \textit{It\^o signature} of  $(Y_t)_{t\in[0,T]}$. We use the notation $\langle \emptyset,\Y_t\rangle=1$ and 
$$\langle I, \Y_t\rangle :=\int_0^t\int_0^{t_{n}}\cdots \int_0^{t_{2}}1dY_{t_1}^{i_1}\cdots dY_{t_n}^{i_n},$$
where $I=(i_1,\ldots,i_n)\in \{1,\ldots,d\}^n,$ in order to define its components. In light of the inner product notation $\langle .,. \rangle$, the quantities $\emptyset$ and $I$ should be interpreted as \textit{selector} vectors (i.e., vectors with a 1 in correspondence with the element to be selected and 0 everywhere else). They extract the \enquote{empty set} component of $\Y_t$ (which is 1 in this formalization) and the generic element $I$, respectively. 
\end{definition} 

To simplify the exposition, we introduce notation.
\begin{definition}\label{def2}
Given $I=(i_1,\ldots,i_n)\in \{1,\ldots,d\}^n$ and $J=(j_1,\ldots,j_m)\in \{1,\ldots,d\}^m,$ we set 
$$|I|:=n,\qquad I':=(i_1,\ldots,i_{n-1})\qquad\text{and}\qquad IJ=(i_1,\ldots,i_n,j_1,\ldots,j_m).$$
For each $k\in\{1,\ldots,d\},$ we  denote by $I(k)$ the number of $k$s in $I$.
We define $|\emptyset|:=0$, $\emptyset':=0$, and $I\emptyset:=\emptyset I:=I$. We also define
$\Ical_n:=\{I\colon |I|\leq n\},$
for each $n\geq0$ and, for $c_I\in \C,$ use the notation
$$\langle  c, \Y_t\rangle:=\sum_{I\in \Ical_n}c_I\langle I, \Y_t\rangle,$$
to signify the inner product between a vector $c$ and the corresponding elements of the signature, i.e., a linear map which associates weights $c_I$ to each element $\langle I, \Y_t\rangle.$ In what follows, vectors of indices will be denoted by a capital letter (i.e., $I,J,H$) and their components by the corresponding lower case letter (i.e., $i_1,j_1,h_1$).
\end{definition}
\begin{remark}\label{aaa}
Observe that, for each $I\neq \emptyset,$ it holds 
$$\langle I, \Y_t\rangle=\int_0^t\langle I',\Y_s\rangle dY^{i_{|I|}}_s.$$
The remark is obvious given that the elements of $\Y_t$ are progressively enlarging iterated integrals. 
\end{remark}

\subsection{Time-extended It\^o signature}\label{exsig}

The \textit{time-extension} of $\Y_t$ includes time as a signal. Rather than work with a generic $\Y_t$, it is now convenient to time-extend the \textit{It\^o signature} of a $d$-dimensional Brownian motion $W:=(W^1,\ldots,W^d)$. In this case, $I=(i_1,\ldots,i_n)\in \{0,\ldots,d\}^n$ and Definition \ref{def2} continues to apply.

Set, as earlier, $\widehat W_t:=(t,W_t)$. We denote the components of $\widehat W_t$ by
$$\widehat W_t^0:=t\qquad\text{and}\qquad\widehat W_t^i:=W_t^i,$$
and the time-extended \textit{It\^o signature} of $W$ by $\widehat\W$. As discussed above, the elements of $\widehat\W$ are the objects that will play a central role in the expansions (before discontinuities are introduced in Section~\ref{Levy}).

\begin{remark}\label{bbbrem}
Given Remark \ref{aaa}, we now have that, for each $I\neq \emptyset,$ it holds that
$$\langle I, \widehat{\W}_t\rangle = \int_0^t\langle {I'},\widehat\W_s\rangle d\widehat{W}^{i_{|I|}}_s = \int_0^t\langle {I'},\widehat\W_s\rangle1_{\{i_{|I|}=0\}} ds+\int_0^t\langle {I'},\widehat\W_s\rangle1_{\{i_{|I|} \neq 0\}} d W^{i_{|I|}}_s.$$
In words, each element of $\widehat\W_t$ has a representation in terms of an It\^o diffusion. The presence of time yields a finite variation (drift) component which is active when the local martingale component is not, and vice-versa.
\end{remark}

Next, we introduce the $\star$ operator. In essence, the $\star$ operator is defined as multiplying elements of $\widehat\W.$ 

\begin{definition}[\textbf{The $\star$ operator}]\label{lem5} For each $I$ and $J,$ define
$$\langle I\star J,\widehat\W_t\rangle : = \langle I,\widehat\W_t\rangle\langle J,\widehat\W_t\rangle.$$
\end{definition}

\noindent The following lemma discusses the fundamental property of the operator, which provides a (recursively defined) map from $I$ and $J$ to $I\star J$. 

\begin{lemma}\label{prop}
Set $\rho_{ij}:=1_{i=j>0},$ then 
$$I\star J=(I'\star J)i_{|I|}+(I\star J')j_{|J|}+\rho_{i_{|I|}j_{|J|}}(I'\star J')(0),$$
and $\emptyset\star I=I\star\emptyset=I$.
\end{lemma}
\noindent \textbf{Proof.} See Appendix \ref{Appendix1}.

\vspace{0.5cm}

It is easily seen that the expression is due to It\^o's Lemma. Because of the centrality of the result in Lemma \ref{prop}, we provide three examples, starting with a rather simple case in which we square the time component of $\widehat\W_t$, i.e., the element $\langle 0,\widehat\W_t \rangle = t.$ 

\begin{example}[\textbf{The case $0\star0$}]
Given Lemma \ref{prop}, we have $0\star0=2(00)$. Thus, 
$$\langle 0,\widehat\W_t\rangle\langle 0,\widehat\W_t\rangle
=\langle 0 \star 0, \W_t\rangle 
=2\langle 00 ,\widehat\W_t\rangle,$$
where, again, the first equality is only definitional. Since $\langle 0,\widehat\W_t\rangle\langle 0,\widehat\W_t\rangle = t^2$ and $2\langle 00 ,\widehat\W_t\rangle = 2 \int_0^tsds,$ we have
$$t^2 = 2 \int_0^tsds,$$ which is obviously verified based on simple calculus (or It\^o's Lemma). 
\end{example}
\begin{example}[\textbf{The case $1\star1$}]\label{uno}
Given Lemma \ref{prop}, we have $1\star 1 = 2(11)+(0).$ Hence,
$$\langle 1,\widehat\W_t\rangle\langle 1,\widehat\W_t\rangle
=\langle 1 \star 1, \W_t\rangle 
=2\langle 11 ,\widehat\W_t\rangle + \langle 0 ,\widehat\W_t\rangle.$$
Since $\langle 1,\widehat\W_t\rangle\langle 1,\widehat\W_t\rangle = W_t^2$, $2\langle 11 ,\widehat\W_t\rangle = 2 \int_0^tW_sdW_s$ and $\langle 0,\widehat\W_t\rangle = t,$ we have
$$W^2_t =2\int_0^t W_sdW_s+t,$$
\end{example}\label{ex4}
\noindent which is just an application of It\^o's Lemma.

\begin{example}[\textbf{The case $11\star1$}]\label{due}
This is a case in which the $\star$ operator multiplies sets of different cardinality. The result is a right-hand side variable which is also expressed in terms of the $\star$ operator itself and, therefore, requires a simple recursion. Given Lemma \ref{prop}, we have $11\star 1=(1\star1)1 + 111 + 10$. This is, of course, the same as $11\star 1=3(111) +01+10$ in light of the previous example. Thus,
 
$$\langle 11,\widehat\W_t\rangle\langle 1,\widehat\W_t\rangle
=\langle 11 \star 1, \W_t\rangle 
=3\langle 111 ,\widehat\W_t\rangle + \langle 01 ,\widehat\W_t\rangle+\langle 10 ,\widehat\W_t\rangle,$$
where, again, the first equality is only definitional. Since $\langle 11,\widehat\W_t\rangle\langle 1,\widehat\W_t\rangle = W_t\int_0^tW_sdW_s$, $\langle 111 ,\widehat\W_t\rangle = \int_0^t \left(\int_0^s W_u dW_u \right)dW_s,$ $\langle 01 ,\widehat\W_t\rangle = \int^t_0sdW_s$ and $\langle 10 ,\widehat\W_t\rangle = \int^t_0W_sds,$ we have
$$W_t\int_0^tW_sdW_s = 3 \int_0^t \left(\int_0^s W_u dW_u \right)dW_s + \int^t_0sdW_s + \int^t_0W_sds,$$ which is, once more, verified based on It\^o's Lemma. 
\end{example}

\begin{remark}
Lemma \ref{prop} is analogous, in our context, to the \textit{shuffle product} routinely used in the signature-based literature in mathematical finance (c.f., e.g., \citealp{CGS:23}). The difference between the star product in Lemma \ref{prop} and the shuffle product (in, e.g., Definition 2.2 of \citealp{CGS:23}) results from our reliance on It\^o integrals, as typically done in continuous-time econometrics, rather than on Stratonovich integrals, which are common in the signature-based literature. 
\end{remark}

\color{black}

In essence, as made clear by Lemma \ref{prop} and the three examples above, the time-extended \textit{It\^o signature} is such that every polynomial in the signature elements may be expressed as a linear function of other elements of the signature. This feature will be central to the nonparametric expansion of moments to any order in the following sections. Before turning to moments, however, we present the expansion of the process for $W:=(W^0,\ldots,W^d)$, the results in Subsection~\ref{logic} being for the case of a scalar $W$ only. In the process, we formalize the properties of the error term $\e_n(t).$  

\subsection{Back to the process expansion}\label{proexp}

Consider a $n+1$ ($W,t$)-differentiable It\^o process $(X_t)_{t\in[0,T]}$ driven by a $d$-dimensional Brownian motion $W:=(W^1,\ldots,W^d).$ Using the logic in Subsection~\ref{logic} and the notational conventions in Subsection~\ref{sig} and Subsection~\ref{exsig}, we have
\begin{equation}\label{exp1}
X_t=\sum_{I\in \Ical^n}c_{I}\langle I,\widehat\W_t\rangle+\e_n(t) =\langle c,\widehat\W_t\rangle+\e_n(t),
\end{equation}
where $(\e_n(t))_{t\in[0,T]}$ is an error term given by
\begin{eqnarray}\label{error}
\e_n(t)=\sum_{|I|=n+1}
 \int_0^t\int_0^{t_{n+1}}\cdots\int_0^{t_{2}} c_{I}(t_1) d\widehat W_{t_{1}}^{i_{1}}\cdots d\widehat W_{t_{n+1}}^{i_{n+1}},
 \end{eqnarray}
for some stochastic process $t\mapsto c_{I}(t)$. 

In order to fully characterize the expansion of the process $X_t$ we now simply have to bound the error term $\e_n(t).$ To this extent, we introduce the following, rather weak, assumption.

\begin{condition}\label{eqn11} For each $I,$ with $|I|=n+1,$ we require the map
\begin{equation}
t\mapsto \E_0[c_{I}(t) ^{2N}]
\end{equation}
to be bounded on $[0,\delta],$ for some $\delta>0$ and $N\in \N$. 
\end{condition}

Next, we provide a probability bound for $\e_n(t).$ We begin with two lemmas. The second lemma will immediately deliver the required result as a corollary.

Recall that, for each $I_1,\ldots,I_n,$ Definition~\ref{lem5} states that
$$\E_0[\langle I_1,\widehat\W_t\rangle\cdots \langle I_n,\widehat\W_t\rangle]
=\E_0[\langle I_1\star\cdots\star I_n,\widehat\W_t\rangle].
$$
Because polynomials in the signature elements can always be written as linear combinations of elements of the signature (from Lemma \ref{prop}), any moment of $\widehat\W_t$ may be expressed as a linear combination of terms of the form $\E_0[\langle I,\widehat\W_t\rangle]$. Given the Gaussian properties of Brownian motion, these moments are particularly easy to express as in the following lemma.

\begin{lemma}\label{lem_moments}
For each $I$, it holds that
$$\E_0[\langle I,\widehat\W_t\rangle]
=
\begin{cases}
\frac{t^n}{n! } & \text{if }I= (0\cdots0), |I|=n\\
0 &\text{else}.
\end{cases}$$
\end{lemma}

\noindent \textbf{Proof.} The claim is easily seen to follow from an application of Fubini's theorem.

\begin{lemma}\label{lemm9}
Fix $K\in \N$, a vector $I:=(i_1,\ldots, i_{|I|})$ and a process $(H_t)_{t\in[0,T]}$ such that $t\mapsto \E_0[H_t^{2K}]$ is bounded on $[0,\delta]$. Then, there exists a constant $C_{2K}>0$ such that 
$$\E_0\bigg[\bigg(\int_0^t\int_0^{t_{|I|}}\cdots\int_0^{t_{2}} H_{t_{1}} d\widehat W_{t_{1}}^{i_{1}}\cdots d\widehat W_{t_{|I|}}^{i_{|I|}}\bigg)^{2K}\bigg]
\leq \frac{C_{2K}^{|I|-I(0)}t^{K(|I|+I(0))}}{|I|!}\sup_{t\in[0,\delta]}\E_0[H_{t}^{2K}],$$
for each $t\in[0,\delta]$.
\end{lemma}
\noindent \textbf{Proof.} See Appendix \ref{Appendix1}.

\begin{corollary}\label{corr1}
For each $t\in [0,\delta]$,  $I$ and $m\geq0,$ it holds
\begin{equation}\label{eqn10}
\E_0[\langle I,\widehat\W_t\rangle^{m}]\leq \frac{C_{2m}^{(|I|-I(0))/2}t^{m(|I|+I(0))/2}}{\sqrt{|I|!}},
\end{equation}
for some constant $C_{2m}>0$. Similarly, fix $N\in \N$ such that Condition \ref{eqn11} is satisfied. Then, for each
$m\leq 2N$ and $t\in [0,\delta],$ it also holds
\begin{equation}\label{eqn9}
\E_0[|\e_n(t)|^{m}]\leq \frac{C_{2N}^{m(n+1)/2N}t^{m(n+1)/2}}{((n+1)!)^{m/2N}},
\end{equation}
for some constant $C_{2N}>0$. 
\end{corollary}
\noindent \textbf{Proof.} See Appendix \ref{Appendix1}.

\vspace{0.2cm}
\begin{remark}\label{rem1}
We emphasize that Eq.~(\ref{eqn10}) covers the case $\E_0\langle \emptyset,\widehat\W_t\rangle$. In this case, in fact, $|\emptyset| =0$ and $I(0) = 0,$ thereby leading to $\E_0\langle \emptyset,\widehat\W_t\rangle \leq C.$ In the proofs, we will - sometimes conservatively - bound $\E_0\langle I,\widehat\W_t\rangle$ by $C$ in order to account, explicitly, for the case $I= \emptyset.$
\end{remark}

We note that the constant $C_{2N}$ in Eq.~(\ref{eqn9}) is a Burkholder-Davis-Gundy constant. It is known that $C_{2N} = C$ if $N = 1/2$ and $C_{2N} = O((2N)^{2N}).$ Setting $m=1$ and $N=1/2$, we have that 
\begin{equation}\label{eqn91}
\E_0[|\e_n(t)|] \leq \frac{C^{n+1}t^{(n+1)/2}}{(n+1)!}.
\end{equation}
Given Stirling's formula, $(n+1)! \sim \sqrt{2\pi (n+1)}\left(\frac{n+1}{e}\right)^{(n+1)},$ which implies that $\frac{C^{n+1}}{(n+1)!} \rightarrow 0,$ as $n\rightarrow \infty.$

We conclude that, by a simple application of Markov's inequality, Corollary \ref{corr1} implies that the error in the general (i.e., for any $n$) process expansion in Eq.~(\ref{exp1}) satisfies
\begin{equation*}
\e_n(t) = O_p(t^{n/2}).
\end{equation*}

\section{Expanding \enquote{regular} conditional moments}\label{regular}

We now turn to the local expansion of moments $\E_0[f(X_t)].$ In order to guarantee sufficient integrability, we will often work with functions in the following set
$$C_p^k(\R):=\{f\in C^k(\R)\colon |f^{(j)}(x)|\leq C (1+|x|^p), \quad\text{ for some } p\in \N,\ \text{for all }j\leq k\}.$$
In this section, the functions have bounded derivatives at time 0. We consider the unbounded case in the next section.

We begin with a lemma which provides a representation of the drift of a process $f(\langle c,\widehat \W_t\rangle)$ (in isolation as well as multiplied by an arbitrary process $\langle d,\widehat \W_t\rangle$) in terms of a linear combination of the \textit{It\^o signature}'s components. The lemma also introduces the coefficients $\Gcal_{.,.}(d),$ which will be helpful to automate our methods. 

\begin{lemma}\label{lem3}
For each $f\in C^2_p(\R),$  and vectors $c,d,$  it holds that
\begin{align*}
f\Big(\langle c,\widehat \W_t\rangle\Big)\langle d,\widehat \W_t\rangle
&=\int_0^t \bigg(\sum_{k=0}^2f^{(k)}\Big(\langle c,\widehat \W_s\rangle\Big)\langle\Gcal_{c,k}(d),\widehat \W_s\rangle \bigg)ds+\text{martingale},
\end{align*}
with 
\begin{align*}
\Gcal_{c,0}(d)&:= \sum_{H\in \Ical_n}d_HH'1_{\{h_{|H|}=0\}},\\
\Gcal_{c,1}(d)&:=\sum_{I,H\in \Ical_n}c_Id_H I'\star (H1_{\{i_{|I|}=0\}}+ {\rho_{i_{|I|},h_{|H|}}}   H' ),\\
\Gcal_{c,2}(d)&:=\sum_{I,J,H\in \Ical_n}c_Ic_Jd_H\frac {\rho_{i_{|I|},j_{|J|}}} 2I'\star J'\star H.
\end{align*}
\end{lemma}
\noindent \textbf{Proof.} See Appendix \ref{Appendix1}.

\vspace{0.3cm}

Next, we expand the conditional (to time 0 information) expectation of the process $f(\langle c,\widehat \W_t\rangle).$ A combination of Taylor's approximation theorem and the representation in Lemma \ref{lem3} will lead to the result.

\begin{theorem}\label{thm1}
For each $f\in C_p^{2N}(\R)$ with $N\in \N$, and a vector $c,$ it holds that
$$\E_0[f(\langle c,\widehat \W_t\rangle)]=f(c_\emptyset)+\sum_{n=1}^N
\frac 1 {n!}\bigg(\sum_{k_1,\ldots,k_n=0}^2f^{(k_1+\ldots+k_n)}(c_\emptyset)\langle\emptyset,\Gcal_{c,k_1,\ldots,k_n}(\emptyset)\rangle\bigg)t^n+o(t^N),$$
where $\Gcal_{c,k_1,\ldots,k_n}:\Ical\to\Ical$ is defined recursively as $\Gcal_{c,k_1,\ldots,k_n}= \Gcal_{c,k_n} \circ \Gcal_{c,k_1,\ldots,k_{n-1}}$ or, more transparently, as $\Gcal_{c,k_1,\ldots,k_n}(I)= \Gcal_{c,k_n}(\Gcal_{c,k_1,\ldots,k_{n-1}}(I))$. 
\end{theorem}
\noindent \textbf{Proof.} See Appendix \ref{Appendix1}.

\vspace{0.3cm}

Importantly, the coefficients in the expansion of $\E_0[f(\langle c,\widehat \W_t\rangle)]$ can be computed explicitly by means of an algorithm. The following remark provides details. 
\begin{remark}\label{remcomp}
Setting 
\begin{align*}
G_0(H)&:=H'1_{\{h_{|H|}=0\}}, \\ 
G_1(I,H)
&:=I'\star \Big(H1_{\{i_{|I|}=0\}}+ {\rho_{i_{|I|},h_{|H|}}}   H'\Big),\\
\textrm{and} ~~~ G_2(I,J,H)
&:=\frac {\rho_{i_{|I|},j_{|J|}}} 2I'\star J'\star H,
\end{align*}
we can write $\Gcal_{c,0}(d)=\sum_{H\in \Ical_n}d_H G_0(H)$, $\Gcal_{c,1}(d)=\sum_{I,H\in \Ical_n}c_Id_H G_1(I,H)$, and $\Gcal_{c,2}(d)=\sum_{I,J,H\in \Ical_n}c_Ic_Jd_H G_k(I,J,H)$. Because each $G_k$ is a multilinear map, we may express compositions of $\Gcal_{c,k}$ as products of the $G_k$s. For example:
$$\Gcal_{c,1,0}(d)=\sum_{I,H\in \Ical_n}c_Id_H G_{0,1}(I,H),$$
where $G_{0,1}(I,H)=\sum_{J\in \Ical_n} \langle J,G_1(I,H)\rangle G_0(J).$\footnote{A code performing these computations is available upon request.} 
 \end{remark}
 
Finally, we turn from an expansion of $\E_0[f(\langle c,\widehat \W_t\rangle)]$ to the expansion of the object of interest, namely $\E_0[f(X_t)].$ We begin with a bound on the error term.
 
\begin{lemma}\label{lem4f}
Consider an $n+1$-times $({W},t)$-differentiable process $(X_t)_{t\in[0,T]}$ with expansion
$$
X_t=\langle c,\widehat\W_t\rangle+\e_n(t),
$$
satisfying Condition \ref{eqn11}. Then, for each $f\in C^1(\R)$ such that $ \textrm{sup}|f'|<\infty,$  it holds
$$\E_0[f(X_t)]=\E_0[f(\langle c,\widehat\W_t\rangle)]+o(t^{n/2}).$$
\end{lemma}
\noindent \textbf{Proof.} See Appendix \ref{Appendix1}.

\vspace{0.3cm}
 
The main result in this section is now easily obtained. We expand the conditional mean of the function $f$ to the order $\lceil n/2\rceil,$ where $n$ is the order of the expansion of the underlying process (as given in Section~\ref{process}). Given Lemma \ref{lem4f}, this is guaranteeing that the approximation error is of smaller order than $t^{n/2}$.
 
\begin{theorem}\label{thm2}
\textbf{(Expanding \enquote{regular} moments.)} Consider an $n+1$-times $( W ,t)$-differentiable process $(X_t)_{t\in[0,T]}$ with expansion
$$
X_t=\langle c,\widehat\W_t\rangle+\e_n(t),
$$
satisfying Condition \ref{eqn11}. Then, for each $f\in C^{n+1}_p(\R)$ with $\sup|f'|<\infty,$ it holds that
\begin{eqnarray}\label{for}
\E_0[f(X_t)]=f(X_0)+\sum_{\ell=1}^{\lceil n/2\rceil}
\frac 1 {\ell!}\bigg(\sum_{k_1,\ldots,k_\ell=0}^2f^{(k_1+\ldots+k_\ell)}(X_0)\langle\emptyset,\Gcal_{c,k_1,\ldots,k_\ell}(\emptyset)\rangle\bigg)t^\ell+o(t^{n/2}).
\end{eqnarray}
\end{theorem}
\noindent \textbf{Proof.} The claim derives from Lemma~\ref{lem4f} and Theorem~\ref{thm1} noting that, in the process expansion, $c_\emptyset=X_0$. 

\vspace{0.3cm}

In essence, the time-0 expectation of the function of the process at $t$ is the function of the process itself at 0 plus a linear combination of its derivatives at 0 and 0-values of suitable coefficients. Derivatives and coefficients are defined recursively (but explicitly) for every level of the expansion. Below, we provide intuition for the recursion.

\subsection{Applications}

We discuss two applications: a first-order expansion of the characteristic function of the squared changes in the $X_t$ process and a general $n$-order expansion of the $k^{th}$-order moment $\E_0[X_t^k].$ The latter is of general interest. The former is simply intended as illustrative. A more general approach to the computation of characteristic function expansions (to any order) will be discussed in the next section. 

\begin{example}\label{ex2}
Fix $n=2$, $u\in \R$ and a $(W,t)$-differentiable process $(X_t)_{t\in[0,T]}$ satisfying Condition~\ref{eqn11}. Theorem~\ref{thm2} yields the following representation:
\begin{align*}
\E_0[e^{iu(X_t-X_0)^2}]=1+iuc_1^2 t+o(t).
\end{align*}
Note, in fact, that for $f(x)=\exp(iu(x-X_0)^2),$ we get $f'(x)=2iu(x-X_0)f(x)$ and $f''(x)=(2iu-4u^2(x-X_0)^2)f(x)$ and hence
\begin{align*}
\E_0[e^{iu(X_t-X_0)^2}]=1+f^{(0)}(X_0)\langle\emptyset,\Gcal_{c,0}(\emptyset)\rangle +f^{(1)}(X_0)\langle\emptyset,\Gcal_{c,1}(\emptyset)\rangle+f^{(2)}(X_0)\langle\emptyset,\Gcal_{c,2}(\emptyset)\rangle +o(t),
\end{align*}
with $f(X_0)=1$, $\langle\emptyset,\Gcal_{c,0}(\emptyset)\rangle = 0$, $f^{(1)}(X_0) =0,$ $f^{(2)}(X_0) = 2iu$ and $\langle\emptyset,\Gcal_{c,2}(\emptyset)\rangle = {c_1^2}/{2}$, which leads to the expansion above. For clarity, we note that $\langle\emptyset,\Gcal_{c,2}(\emptyset)\rangle$, i.e., the first element of the vector $\Gcal_{c,2}(\emptyset),$ is - in this example - the term $\langle\emptyset, c_1c_1\frac{\rho_{1,1}}{2}\emptyset \star \emptyset \star \emptyset \rangle= \langle\emptyset,\frac{c^2_1}{2} \emptyset \rangle = \frac{c^2_1}{2}.$ 
\end{example}

We now turn to the $k^{th}$-order moment $\E_0[X_t^k].$ Before stating the main result, we introduce a helpful lemma which clarifies how the bound on the error term extends to the current example. 

\begin{lemma}\label{lem4}
Consider an $n+1$-times $( W, t)$-differentiable process $(X_t)_{t\in[0,T]}$ with expansion
$$
X_t=\langle c,\widehat\W_t\rangle+\e_n(t),
$$
and fix $N\in \N$ such that Condition \ref{eqn11} is satisfied.
Then, for each $k<2N,$ it holds that
$$\E_0[X_t^k]=\E_0[\langle c,\widehat\W_t\rangle^k]+o(t^{n/2}).$$
\end{lemma}

\noindent \textbf{Proof.} See Appendix \ref{Appendix1}.

\vspace{0.3cm}
 
\begin{example}[\textbf{$\mathbf{k^{th}}$ moment}]\label{exam}
Consider an $n+1$-times $( W,t)$-differentiable process $(X_t)_{t\in[0,T]}$ with expansion
$$
X_t=\langle c,\widehat\W_t\rangle+\e_n(t),
$$
satisfying Condition~\ref{eqn11} for some $N\in \N$.
Then, for each $k<2N,$ it holds
$$\E_0[X_t^k]=X_0^k+\sum_{\ell=1}^{\lceil n/2\rceil}
\frac 1 {\ell!}\bigg(\sum_{k_1,\ldots,k_\ell=0}^2\frac {k!}{(k-(k_1+\ldots+k_\ell))!}X_0^{k-(k_1+\ldots+k_\ell)}\langle\emptyset,\Gcal_{c,k_1,\ldots,k_\ell}(\emptyset)\rangle\bigg)t^\ell+o(t^{n/2}).$$
\end{example}

\noindent \textbf{Proof.} Since $c_\emptyset=X_0$ and the $m$-th derivative of $x^k$ is $\frac {k!}{(k-m)!}x^{k-m},$ the claim follows by Lemma~\ref{lem4} and Theorem~\ref{thm1}.

\subsection{The logic of the moment expansion}

Consider the following (sufficiently-differentiable) homogeneous Markov  diffusion process
$$dX_t=\mu(X_t)dt+\sigma(X_t)dW_t,$$
and a (sufficiently-differentiable  and integrable) function $f$. By It\^o's Lemma, we have
$$f(X_t)=\int_0^t \Gcal f (X_s) ds +\text{martingale},$$
where $\Gcal f(x)=\mu(x)f'(x)+\frac 1 2 \sigma^2(x)f''(x)$ is the process' infinitesimal generator. Because the infinitesimal generator may be viewed as the derivative of the time-0 conditional expectation of the function $f$, we have that
\begin{eqnarray*}
\frac {d^n}{dt^n}\E_0[f(X_t)]=\frac {d^{n-1}}{dt^{n-1}}\E_0[ \Gcal f (X_t)]
=\frac {d^{n-2}}{dt^{n-2}}\E_0[ (\Gcal\Gcal f )(X_t)]
\ldots=\E_0[ \underbrace{(\Gcal\cdots\Gcal}_{n \text{ times}} f )(X_t)].
\end{eqnarray*}
By Taylor's approximation theorem, we conclude that
$$ \E_0[f(X_t)]=f(X_0)+ \Gcal f (X_0) t+ \Gcal\Gcal f(X_0)\frac {t^2} 2 +\ldots+\underbrace{\Gcal\cdots\Gcal}_{n\text{ times}}f(Y_0)\frac {t^n} {n!}+o(t^n).$$
Because of Remark \ref{bbbrem}, we know that the process given by all the components of the signature of depth less than $n$ may be expressed as a homogeneous Markov diffusion process. Thus, using the notation $f_c:=f(\langle c,\cdot \rangle),$ we have
$$\E_0[f_c(\widehat\W_t)]={f_c(\widehat\W_0)+ \Gcal f_c (\widehat\W_0) t+ \Gcal\Gcal f_c(\widehat\W_0)\frac {t^2} 2 +\ldots+\underbrace{\Gcal\cdots\Gcal}_{n/2\text{ times}}f_c(\widehat\W_0)\frac {t^{n/2}} {(n/2)!}}+o(t^{n/2}),$$
with $\Gcal$ denoting the infinitesimal generator of $\widehat\W$. This latter formula justifies the recursive nature of the derivatives and coefficients in Eq.~(\ref{for}). 
\section{Expanding \enquote{irregular} conditional moments}\label{irregular}
The objective of this section is to derive expansions (to any order) of the characteristic function of the standardized process $\frac{X_t - X_0-c_0t}{c_1\sqrt t},$ i.e., 
$$\E_0\Big[\exp\Big(iu \Big(\frac{X_t - X_0-c_0t}{c_1\sqrt t}\Big)\Big)\Big].$$
We, therefore, focus on a specific irregular moment that has drawn attention in the recent continuous-time econometrics literature. 

In simple examples, characteristic function expansions may be found by using known properties of It\^o semimartingales. We provide one such example next. More generally, this is not the case. The general case will be our focus. 
 \begin{example}
Set $\mu(t)=\gamma W_t$
and $\sigma(t)= \alpha.$ Write 
\begin{align*}
X_t&=X_0+\int_0^t \mu(s) ds +\int_0^t \sigma(s) dW_s.
\end{align*}
By It\^o's Lemma, we have $\int_0^t W_s ds=W_tt-\int_0^t s dW_s=\int_0^t (t-s) dW_s.$ We may, therefore, write the following:
\begin{align*}
\E_0[\exp(iu(X_t - X_0)/(\alpha \sqrt t))]
&=
\E_0[\exp((\gamma\int_0^tW_sds+\alpha W_t)iu/(\alpha \sqrt t))]\\
&=\E_0[\exp((\int_0^t(\alpha+\gamma (t-s))dW_s)iu/(\alpha \sqrt t))]\\
&=\exp(-\frac{\sigma_t^2u^2}{2\alpha^2t})\\
&=
\exp(-u^2\gamma^2\frac {t^2} {\alpha^26}-u^2\frac\gamma\alpha \frac t 2 -\frac{u^2}2)\\
&=\exp(-\frac{u^2}2)(1-u^2\frac\gamma\alpha \frac t 2+o(t)),
\end{align*}
where in the third equality we have used the fact that $\int_0^t(\alpha+\gamma (t-s))dW_s$ is a normal random variable with conditional variance process
$$\sigma_t^2:=\E_0[(\int_0^t(\alpha+\gamma (t-s))dW_s)^2]=\int_0^t(\alpha+\gamma s)^2ds=\frac 1 {3\gamma} ((\alpha+\gamma t)^3-\alpha^3)=\gamma^2\frac {t^3} 3+\alpha\gamma t^2+\alpha^2t.$$
\end{example}

This simple example yields a local approximation to the Gaussian characteristic function of the (suitably-standardized) driving Brownian motion. The expansion is of first order in $t$ and depends on the features of the drift process, which is also a function of the assumed driving Brownian motion. 

This sort of expansions are rather complex in general continuous-time càdlàg models (c.f., e.g., \citealp{BR:17}, and \citealp{To:21}). Even low-order expansions require involved computations, something which limits their applicability beyond short horizons. The $W$-transforms in the work of \cite{BR:17}, for instance, have been shown to be helpful as a general tool aiding derivations but do not immediately lead to automated expansions to any order. Automation to arbitrary (high) orders is, instead, our objective in this article. To this extent, we derive a simple algorithmic representation which will facilitate asymptotic approximations in which the order of the expansion is $t^n,$ with $t<1$, and $n \rightarrow \infty$. The typical asymptotic approximations are, instead, for a fixed smallish $n$ ($\frac{1}{2}$ and 1, at the most) and $t \rightarrow 0.$ 

Before stating our central result in this section, we provide an auxiliary lemma which introduces the analogue, in our framework, of the $W$-transforms in \cite{BR:17}. Here, by analogy with that work, we utilize $W$-transforms of the \textit{It\^o signature} elements, for which we provide a rather intuitive representation. 

\begin{lemma}\label{lem2}
For each $I\in\Ical^k,$ it holds that
$$\E_0[\langle I,\widehat\W_t\rangle\exp(iuW_t^1)]
=
\begin{cases}
(iu)^{I(1)}\frac {t^k}{k! }\exp(-u^2t/2) & \text{if }I=\{0,1\}^k,\\
0 &\text{else},
\end{cases}$$
where we recall (c.f. Definition \ref{def2}) that $I(1)$ is the number of ones in $I$. 
\end{lemma}
\noindent \textbf{Proof.} See Appendix \ref{Appendix1}.

\vspace{0.3cm}

Importantly, $W$-transforms of \textit{products} of the \textit{It\^o signature} elements can be easily derived by employing the key property of the $\star$ operator, as implied by Lemma \ref{prop}: products are linear maps of other elements of the \textit{It\^o signature}. Below, we provide an example.

\begin{example}

We begin with the $W$-transform of $\langle 1 \star 1, \W_t\rangle,$ as in Example \ref{uno}:
\begin{align*}
\E_0[\langle 1 \star 1, \W_t\rangle \exp(iuW^1_t)]&=\E_0[(2\langle 11, \widehat \W_t\rangle+\langle 0, \widehat \W_t\rangle) \exp(iuW^1_t)]
=\exp(u^2t/2)\left(2(iu)^{2}\frac {t^2}{2! }+ t\right).
\end{align*}
We now turn to the $W$-transform of $\langle 1 \star 11, \W_t\rangle,$ as in Example \ref{due}:
\begin{align*}
\E_0[\langle 1 \star 11, \W_t\rangle \exp(iuW^1_t)]&=\E_0[(3\langle 111, \widehat \W_t\rangle+\langle 10+01, \widehat \W_t\rangle) \exp(iuW^1_t)]\\
&=\exp(u^2t/2)\left(3(iu)^{3}\frac {t^3}{3! }+2(iu) \frac{t^2}2\right).
\end{align*}

\end{example}
\noindent Next, we give details on the process of interest. 

\begin{condition}\label{processcond}
Fix $m\geq 2.$ Let $(X_t)_{t\in[0,T]}$ be an $m+2$-times $(W ,t)$-differentiable process  with expansion
$$
X_t=\langle c,\widehat\W_t\rangle+\e_{m+1}(t),
$$
satisfying Condition~\ref{eqn11}.
Suppose that  $X_0=c_0=0$ and that $c_{Ii}=0$ for each $i>1$,  $c_{Ii1}=0$ for each $i>2$ and $c_{Ii0}=0$ for each $i>3.$
\end{condition}

\noindent The following process and it first two layers, in particular, satisfy Condition \ref{processcond} whenever sufficiently differentiable: 
\begin{align*}
X_t&=\int_0^t c_0(s) ds +\int_0^t c_1(s) dW_s^1\\
c_1(t)&=c_1(0)+\int_0^t c_{01}(s) ds +\int_0^t c_{11}(s) dW_s^1+\int_0^t c_{21}(s) dW_s^2\\
c_0(t)&=\int_0^t c_{00}(s) ds +\int_0^t c_{10}(s) dW_s^1+\int_0^t c_{20}(s) dW_s^2+\int_0^t c_{30}(s) dW_s^3.
\end{align*}

The drift process, $c_0(t),$ and the volatility process, $c_1(t),$ are allowed to be correlated with the level process, $X_t$, and with themselves. The correlation is, however, not one because of the dependence of all three processes on a cascade of independent Brownian motions. Imposing $X_0=c_0=0$ is, as we note below, solely for notational simplicity and without loss of generality. 

We are now ready to state the central result in this section. Again, this result can be implemented through an algorithm. 
\begin{theorem}\label{prop2}
\textbf{(Expanding the characteristic function of the standardized process.)}
Fix $u\in \R$ and consider a process $(X_t)_{t\in [0,T]}$ satisfying Condition~\ref{processcond}. Then,
\begin{align*}
&\E_0\Big[\exp\Big(iu\frac{X_t}{c_1\sqrt t}\Big)\Big]
=
\E_0\bigg[ \exp\Big(\frac{iu}{\sqrt t}{W_t^1}\Big)\bigg]\\
&\quad+
\sum_{\ell=1}^m
 \frac{(iu)^{\ell}}{ t^{\ell/2} (\ell!)}\sum_{I_1,\ldots,I_\ell\in \Ical^{m+1}}
 \overline c_{I_1}\cdots  \overline c_{I_\ell}\E_0\bigg[\langle I_1\star \cdots\star I_\ell,\widehat\W_t\rangle \exp\Big(iu\frac{W_t^1}{\sqrt t}\Big)\bigg] +o(t^{m/2} ),
\end{align*}
where $\bar c_1=0$, and $\bar c_I=c_I/c_1$ for $I\neq (1)$.
\end{theorem}
\noindent \textbf{Proof.} See Appendix \ref{Appendix1}.

\vspace{0.3cm}

\noindent As a direct application of Theorem~\ref{prop2}, we recover Corollary 3 of Theorem 1 in \cite{BR:17}. We observe that, if $c_0\neq0$ and/or $X_0\neq0,$ the same result would hold with $X_t$ replaced by $X_t-X_0-c_0t$.
\begin{corollary}\label{BR}
\textbf{(A second-order (in $\sqrt{t}$) expansion.)}
 Fix $u\in \R$ and consider a process $(X_t)_{t\in [0,T]}$ satisfying Condition~\ref{processcond}.
Then
\begin{align*}
&\E_0\Big[\exp\Big(iu\frac{X_t}{c_1\sqrt t}\Big)\Big]\\
&=e^{-\frac{u^2} 2}\left(1+\Big[- \frac{c_{11}}{c_1}\frac{iu^3}2\Big]{\sqrt t}\right.\\
&\qquad\left.+\frac 1 2\Big [-\Big(\frac{ c_{01}}{c_1}+\frac{c_{10}}{c_1}\Big)u^2
+\Big(\frac{ c_{11}}{c_1}\Big)^2\Big(-\frac 1 2 u^2+u^4-\frac{1}{4}u^6\Big)
+\Big(\frac{ c_{21}}{c_1}\Big)^2\Big(\frac {1}{3}u^4-\frac {1}{2}u^2\Big)+\frac{ c_{111}}{c_1}\frac{u^4}{3} \Big]t \right)\\
&+o(t).
\end{align*}
\end{corollary}

\noindent \textbf{Proof.} See Appendix \ref{Appendix1}.

\vspace{0.3cm}

This is a second order (in $\sqrt{t}$) expansion of the characteristic function of the \textit{standardized} process. Because the standardized process is locally Gaussian, the expansion is - as expected - around the Gaussian characteristic function. The first order yields a skewness adjustment (through the term $c_{11}$ representing the volatility of volatility associated with the driving Brownian motion $W_1$ of the level process). It is clear that the coefficient $c_{11}$ captures the correlation between the level process and the volatility process. The second order yields, instead, adjustments to the second, the fourth and the sixth moment of the distribution of the standardized process. Both the third moment correction and the corrections to positive moments are driven by the dynamics of the volatility process. \cite{BR:17} provide more discussion.  

\begin{remark}
The representation in Theorem~\ref{prop2} may also be used for regular moments. The following Corollary applies it to $\E_0[f(X_t)]$ given a sufficiently-differentiable function $f$ with bounded derivatives at zero. While working with moments of components of the \textit{It\^o signature}, as in Theorem~\ref{prop2}, may be more intuitive than working with iterated coefficients, as in Theorem~\ref{thm2}, both expansions can be easily implemented, given their algorithmic structure. In addition, as we will see in the following section, the expansion in Theorem~\ref{thm2} appears to be better suited to handle discontinuities, an observation which justifies our emphasis on it in Section~\ref{regular}.  
\end{remark}

\begin{corollary}\label{prop3}
\textbf{(Back to \enquote{regular} moments.)}
Consider an $m+1$-times $( W,t)$-differentiable process $(X_t)_{t\in[0,T]}$ with $X_0=0$ and expansion
$$
X_t=\langle c,\widehat\W_t\rangle+\e_m(t),
$$
satisfying Condition~\ref{eqn11} for some $N\in \N$.
 Assume that $f\in C^{m +1}(\R)$ satisfies $\sup|f'|<\infty$ and $\sup|f^{(m+1)}|<\infty$. Then,
\begin{eqnarray*}
\E_0[f(X_t)]
&=&f(0)+\sum_{\ell=1}^m
 \frac{f^{(\ell)}( 0)}{\ell!} \sum_{I_1,\ldots,I_\ell\in \Ical^{m}}
  c_{I_1}\cdots   c_{I_\ell}\E_0\bigg[\langle I_1\star \cdots\star I_\ell,\widehat\W_t\rangle \bigg] +o(t^{m/2} ),
\end{eqnarray*}
\end{corollary}

\noindent \textbf{Proof.} The method of proof leading to the statement of Theorem~\ref{prop2} immediately yields the result.

\vspace{0.3cm}

\begin{example}\textbf{\textit{(Back to Example \ref{ex2}.)}}
Recall that we wish to show that
\begin{align*}
\E_0[e^{iu(X_t-X_0)^2}]=1+iuc_1^2 t+o(t).
\end{align*}
Write $Y_t=X_t-X_0$ and $f(x)=\exp(iux^2).$ Thus, $f'(x)=2iuxf(x)$ and $f''(x)=(2iu-4u^2x^2)f(x)$. Set $m=2.$
Note, in fact, that
\begin{eqnarray*}
\E_0[e^{iuY_t^2}]&=&f(0) +f^{(1)}(0)\sum_{I \in \Ical^{\red 2}}
  c_{I}\E_0\bigg[\langle I,\widehat\W_t\rangle \bigg] \\
  &&\qquad+\frac{f^{(2)}(0)}{2}\sum_{I_1, I_2 \in \Ical^{\red 2}}
  c_{I_1}  c_{I_2}\E_0\bigg[\langle I_1\star I_2,\widehat\W_t\rangle \bigg]+o(t) \\
&=& 1 + iu\sum_{I_1, I_2 \in \Ical^{\red 2}}
  c_{I_1}  c_{I_2}\E_0\bigg[\langle I_1\star I_2,\widehat\W_t\rangle \bigg]+o(t) \\
& = & 1 + iu c^2_{1} \E_0\bigg[\langle 1\star 1,\widehat\W_t\rangle \bigg]+o(t) \\ 
\textrm{given Example \ref{ex4}} & = & 1 + iu c^2_1\underbrace{\E_0\bigg[2\langle 11 ,\widehat\W_t\rangle \bigg]}_{= 0} + iu c^2_1\underbrace{\E_0 \bigg[\langle 0 ,\widehat\W_t\rangle \bigg]}_{= t}+o(t) \\
& = &iu c^2_{1}t + o(t),   
\end{eqnarray*}
  which proves the statement. 
\end{example}

\section{Application: automation to any order}\label{any}

As emphasized, by using the properties of the \textit{It\^o signature}, the reported expansions can be conducted to any order. In this section, we present an additional corollary to Theorem~\ref{prop2} giving the expansion of the characteristic function of the standardized process to the \textit{third order} (in $\sqrt{t}$).    

We assume the same process as in Condition \ref{processcond}. We solely restrict $c_{20}$ to be zero. The implication is that the drift process, $c_0(t),$ and the volatility process, $c_1(t)$, are allowed to be correlated with the level process, $X_t,$ but they are not correlated with each other. This is the model used in \citet{BFR} to price options with expirations within the day. We note that setting $c_{20}=0$ is without loss of generality. The change does not affect the first two orders of the expansion (Corollary \ref{BR} would be unchanged). While it would affect the third order, the corresponding modifications are trivial to incorporate in light of the algorithmic nature of the signature properties. 

It is, indeed, the algorithmic nature of the signature properties which allows us to easily compute the result in the following corollary (in which the order is $t^{3/2}$) as well as results for any order $t^n$. Due to the complexity of these expansions, the cases $n=1/2$ (c.f., \citealp{jacod2014efficient}, and \citealp{todorov2019nonparametric}) and $n=1$ (c.f., \citealp{BR:17}, and \citealp{To:21}) are, to date and to the best of our knowledge, the only ones in the literature. Corollary \ref{BR1}, below, is explicit about the extension to the next order (i.e., the case $n=3/2$). Theorem~\ref{prop2} generalizes to any order.\footnote{The code is available upon request.} 

\begin{corollary}\label{BR1}
\textbf{(A third-order (in $\sqrt{t}$) expansion.)}
Fix $u\in \R$ and consider a process $(X_t)_{t\in [0,T]}$ satisfying Condition~\ref{processcond} for $m=3$. Then, it holds that 
\begin{eqnarray*}
&&\E_0\Big[\exp\Big(iu\frac{X_t}{c_1\sqrt t}\Big)\Big]\\
&=&e^{-\frac{u^2} 2}\left(1+\Big[- \frac{c_{11}}{c_1}\frac{iu^3}2\Big]{\sqrt t}\right.\\
&+&\left.\frac 1 2\Big [-\Big(\frac{ c_{01}}{c_1}+\frac{ c_{10}}{c_1}\Big)u^2
+\Big(\frac{ c_{11}}{c_1}\Big)^2\Big(-\frac 1 2 u^2+u^4-\frac{1}{4}u^6\Big)
+\Big(\frac{ c_{21}}{c_1}\Big)^2\Big(\frac {1}{3}u^4-\frac {1}{2}u^2\Big)+\frac{ c_{111}}{c_1}\frac{u^4}{3} \Big]t \right)\\
&+ & \left [ \frac{iu^3}{2}\left( - \frac{c_{011}}{3 c_1} - \frac{c_{101}}{3 c_1} - \frac{c_{10} c_{11}}{c_1^2} -\frac{c_{110}}{3 c_1} - \frac{c_{01} c_{11}}{c_1^2} - \frac{c_{11}^3}{3 c_1^3} - \frac{c_{11} c_{111}}{c_1^2} - \frac{c_{121} c_{21}}{3 c_1^2} \right) \right. \\
& + & \left. \frac{iu^5}{2}\left(\frac{c_{10} c_{11}}{2 c_1^2} + \frac{c_{01} c_{11}}{2 c_1^2} + \frac{5 c_{11}^3}{4 c_1^3} + \frac{c_{11} c_{111}}{c_1^2} + \frac{c_{1111}}{12 c_1} + \frac{c_{121} c_{21}}{6 c_1^2} + \frac{11c_{11} c_{21}^2}{12 c_1^3} \right) \right.\\
 &+ & \left. \left. \frac{iu^7}{4}\left(- \frac{c_{11}^3}{c_1^3} - \frac{c_{11} c_{111}}{3 c_1^2} - \frac{c_{11} c_{21}^2}{3 c_1^3}\right) + \frac{iu^9}{4}\frac{c_{11}^3}{12 c_1^3}\right]t^{3/2} \right) \\
&+&o(t^{3/2}).
\end{eqnarray*}
\end{corollary}

\noindent \textbf{Proof.} The result is an implication of Theorem~\ref{prop2}.

\vspace{0.3cm}

\noindent In a second-order expansion, like the one in Corollary \ref{BR}, the dynamics of volatility - through leverage and the volatility and volatility - are free to tilt the Gaussian density and give it both skewness and kurtosis. The volatility of volatility is, however, effectively \enquote{frozen}. The third-order expansion in Corollary \ref{BR1} would \enquote{unfreeze} the volatility of volatility (as well as other quantities) by allowing characteristics that drive its dynamics, like $c_{011}$ and $c_{111}$, to contribute to the distributional tilts. This increased depth may speak to risk premia on higher-order equity characteristics, something that has drawn recent attention in financial markets. The VVIX, as an example, was introduced by the CBOE in 2012.

\section{Adding discontinuities}\label{Levy}
In this section, we add a vector $N$ of $e$ independent compound Poisson processes to the original source of randomness (the $d$-dimensional Brownian motion $W:=(W^1,\ldots,W^d)$). The assumed compound Poisson processes have intensities $\lambda_1,\ldots,\lambda_e$ and jump densities $\nu_1,\ldots, \nu_e$. These processes are independent of $W$ and we assume that their Lévy measures satisfy $\lambda_j\int |\xi|^k\nu_j(d\xi)<\infty,$ for each $k$ and $j$. We set $$ Z_t:=(N_t^e,\ldots, N_t^1,t,W_t^1,\ldots,W_t^d).$$ 
We denote the components of $Z_t$ by
$$ Z_t^{-j}:=N_t^j,\qquad  Z_t^0 :=t,\qquad\text{and}\qquad  Z_t^i :=W^i_t,$$
for each $j\in\{1,\ldots,e\}$ and $i\in \{1,\ldots,d\}.$ The \textit{It\^o signature} of $Z_t$ is, instead, denoted by $\Z$.

It will be convenient to work with an extended version of $Z_t$ ($\overline Z_t$). To this extent, for some fixed $m\in \N,$ we consider  
$$ \overline Z_t:=(\overline N_t^{e\cdot m},\ldots, \overline N_t^1,t,W_t^1,\ldots,W_t^d),$$
where
$ \overline N_t:=(N_t^{1,1},\ldots,N_t^{1,m} ,\ldots, N_t^{e,1},\ldots,N_t^{e,m}).$ Thus, e.g., $\overline N^1_t = N_t^{1,1}$ and $\overline N^m_t = N_t^{1,m}.$ Here, $N^{j,k}$ denotes the compound Poisson process with intensity $\lambda_j$ and jump compensator $\lambda_j \int \xi^k\nu_j(d\xi)$: 
$$N_t^{j,k}=\sum_{s\leq t}(\Delta N_t^j)^k.$$ 
Once more, we denote the components of $\overline Z_t$ by
$$ \overline Z_t^{-j}:=\overline N_t^j,\qquad  \overline Z_t^0:=t\qquad\text{and}\qquad  \overline Z_t^i:=W^i_t,$$
for each $j\in\{1,\ldots,e\cdot m \}$ and $i\in \{1,\ldots,d\}$
and its \textit{It\^o signature} by $\overline \Z$.

For future convenience in the notation, we also introduce the bijection
$$\Lcal:\{-e\cdot m,\ldots,-1\}\to \{1,\ldots,e\}\times \{1,\ldots,m\},$$ which gives
$$\overline Z_t^j=N_t^{\Lcal_1(j),\Lcal_2(j)}\qquad\text{and}\qquad \overline Z_t^{\Lcal^{-1}(j,k)}=N_t^{j,k},$$
 for each $j<0$. In words, we can map every element of the signature (defined by $j$) into the corresponding Poisson process (defined by $\Lcal_1(j)$) and the corresponding power (defined by $\Lcal_2(j)$). Fixing, for instance, $m=2$ and $e=3,$ we have
 $$(Z^{-1}_t,\ldots, Z^{-6}_t)=(N_t^{1,1},N_t^{1,2},N_t^{2,1},\ldots,N_t^{3,2}),$$
which implies, e.g., that $Z^{-5}_t = N_t^{3,1}.$

As a first step towards a theory which allows for the addition of compound Poisson processes, we now extend the product operator in previous sections. We do so by introducing the $\ostar$ product.

\subsection{The  $\ostar$ product}
The following Definition and Lemma are the analogue of Definition \ref{lem5} and Lemma \ref{prop} in Section~\ref{process}. 

\begin{definition}\label{lem5new} For each $I$ and $J,$ assume that
\begin{eqnarray} \label{m}
(\Lcal_2(i_k)+\Lcal_2(j_\ell))1_{\{\Lcal_1(i_k)=\Lcal_1(j_\ell)\}}\leq m,
\end{eqnarray}
for each $k\leq |I|$ and $\ell\leq |J|$. We define
$$\langle I\ostar J,\overline \Z_t\rangle {\red :}= \langle I,\overline\Z_t\rangle\langle J,\overline\Z_t\rangle.$$
\end{definition}

\begin{lemma}\label{newprop}
Set $\overline \rho_{ij}:=1_{\{i=j>0\}}$,  $\tau_{ij}:=1_{\{i,j<0\}}1_{\{\Lcal_1(i)=\Lcal_1(j)\}}$ and
\begin{eqnarray}\label{circ}
i\circ j:=\Lcal^{-1}(\Lcal_1(i),\Lcal_2(i)+\Lcal_2(j)).
\end{eqnarray}
Then, 
$$I\ostar J=(I'\ostar J)i_{|I|}+(I\ostar J')j_{|J|}+\overline \rho_{ij}(I'\ostar J')0+\tau_{ij}(I'\ostar J')i_{|I|}\circ j_{|J|},$$
and  $\emptyset\ostar I=I\ostar\emptyset=I$.
\end{lemma}

\noindent \textbf{Proof.} See Appendix \ref{Appendix2}.

\vspace{0.3cm}

\noindent Lemma \ref{newprop} is the result of an application of It\^o's Lemma for discontinuous processes. A simple comparison between Lemma \ref{prop} and Lemma \ref{newprop} clarifies that the difference between the operator $\star$ and the operator $\ostar$ resides in the last term, i.e., the term that captures the impact of discontinuities on the product of \textit{It\^o signature} elements. This is, in essence, a term representing jump co-variation. Next, we provide an illustrative example.

\begin{example}
We recall that $\overline Z^{-1}_t=\overline N_t^{1,1}$ and $\overline Z^{-2}_t=\overline N_t^{1,2}$. We have that
$$(-1)\circ(-1)=\Lcal^{-1}(\Lcal_1(-1),\Lcal_2(-1)+\Lcal_2(-1))=\Lcal^{-1}(1,2)=(-2),$$
where the first equality is due to the definition in Eq.~(\ref{circ}) of Lemma \ref{newprop} and the remaining two are due to an application of the bijection introduced in the previous subsection.
As a result, we have that $(-1)\overline \star (-1) = 2(-1-1)+(-2)$. We can now easily verify the relation among indices implied by the $\ostar$ operator by using It\^o's formula:
$$\langle (-1),\overline\Z_t\rangle\langle (-1),\overline\Z_t\rangle
= N^1_t \fdot N^1_t
=2\int_0^t N_s^1dN_s^1+\sum_{s\leq t}(\Delta N^1_s)^2
=2\langle (-1-1),\overline\Z_t\rangle+\langle (-2) ,\overline\Z_t\rangle.$$
Importantly, the example clarifies that, for $m$ large enough (a condition which justifies Eq. (\ref{m})) products of \textit{It\^o signature} elements can be expressed - like in the continuous case - as linear combinations of other \textit{It\^o signature} elements. In the example, the product of the first Poisson component of the signature ($I=(-1)$) is expressed as 2 times the component $I = (-1,-1)$ (the integral of the component $I=(-1)$ with respect to itself) and the component $I = (-2) $ (the quadratic variation of the component $I=(-1)$). It also clarifies the motivation for working with $\overline Z$ instead of $Z$.
\end{example}

\subsection{A triplet of expansions}
We now turn to expansions of the discountinuous process $X_t$ as well as to expansions of its \enquote{regular} and \enquote{irregular} moments. 
\subsubsection{The process}
The definition of an $n+1$-times $\overline Z$-differentiable process is consistent with that in Subsection~\ref{proexp}. 
 \begin{definition}
A process $(X_t)_{t\in[0,T]}$ is $n+1$-times $\overline Z$-differentiable if
$$
X_t=\sum_{I\in \Ical^n}c_{I}\langle I,\overline\Z_t\rangle+\e_n(t),
$$
where $(\e_n(t))_{t\in[0,T]}$ is an error term given by
$$\e_n(t)=\sum_{|I|=n+1}
 \int_0^t\int_0^{t_{n+1}}\cdots\int_0^{t_{2}} c_{I}(t_1) d\overline Z_{t_{1}}^{i_{1}}\cdots d\overline Z_{t_{|I|}}^{i_{|I|}},$$
for some stochastic process $t\mapsto c_{I}(t)$. We assume that the map $t\mapsto \E_0[c_{I}(t) ^{2N}] $ satisfies Condition \ref{eqn11}. 
\end{definition}
\begin{remark}
Using the short-hand notation introduced in Definition~\ref{def2}, an $n+1$-times $\overline Z$-differentiable process $(X_t)_{t\in[0,T]}$ may, also, be written compactly as follows: 
\begin{eqnarray}\label{process2}
X_t=\langle c,\overline\Z_t\rangle+\e_n(t).
\end{eqnarray}
\end{remark}
\noindent Next, we turn to a probability bound on the error term. Recall that, for each $I_1,\ldots,I_n,$ by Definition~\ref{lem5new}, it holds that
$$\E_0[\langle I_1\ostar\cdots\ostar I_n,\overline\Z_t\rangle] = \E_0[\langle I_1,\overline\Z_t\rangle\cdots \langle I_n,\overline\Z_t\rangle].
$$
Any moment of $\overline\Z_t$ can, therefore, be expressed as a linear combination of terms of the form $\E_0[\langle I,\overline\Z_t\rangle]$.

\begin{lemma}\label{lem_momentsJ}
For each $I,$ we have that
$$\E_0[\langle I,\overline\Z_t\rangle]
=\begin{cases}
\frac{t^{|I|}}{{|I|}! } \prod_{k=1}^{|I|}\Big( \lambda_{\Lcal_1(i_k)}\int \xi^{\Lcal_2(i_k)}\nu_{\Lcal_1(i_k)}(d\xi)1_{\{i_k<0\}}+1_{\{i_k=0\}}\Big) & \text{if }I=(0,-1,...,-e \cdot m)^{|I|}\\
0 &\text{else}.
\end{cases}$$
\end{lemma}

\noindent \textbf{Proof.} See Appendix \ref{Appendix2}.

\vspace{0.3cm}

Before stating the next result, we extend Definition~\ref{def2}. For a vector $I,$ we denote by $I(<0)$, respectively $I(>0)$, the number of strictly negative, respectively strictly positive, entries in $I$.
\begin{lemma}\label{lem9new}
Fix $ \delta>0$, a vector $I=(i_1,\ldots, i_{|I|})$ and a process $(H_t)_{t\in[0,T]}$  such that $t\mapsto \E_0[H_t^{2K}]$ is bounded on $[0,\delta]$. Then, there exists a constant $C_{2K}>0$ such that 
$$\E_0\bigg[\bigg(\int_0^t\int_0^{t_{|I|}}\cdots\int_0^{t_{2}} H_{t_{1}} d\overline Z_{t_{1}}^{i_{1}}\cdots d\overline Z_{t_{|I|}}^{i_{|I|}}\bigg)^{2K}\bigg]
\leq \frac{C_{2K}^{|I|-I(0)}t^{KI(>0)+2KI(0)+I(<0)}}{|I|!}\sup_{t\in[0,\delta]}\E_0[H_{t}^{2K}]$$
for each $t\in[0,\delta]$.
\end{lemma}

\noindent \textbf{Proof.} See Appendix \ref{Appendix2}.

\vspace{0.6cm}

Differently from the bound in Lemma \ref{lemm9}, it is important to notice that the order of the bound in Lemma \ref{lem9new} is not, in general, strictly increasing in $K$ for each $I$. If $I$ only contains negative entries, i.e. each $\overline Z^i$ is a compound Poisson process, this is - in fact - not the case. 

\vspace{0.3cm}
\begin{corollary}\label{cor1new}
For each $t\in [0,\delta]$,  $I$ and $m\geq0,$ it holds that
\begin{equation}
\E_0[\langle I,\widehat\Z_t\rangle^{m}]\leq \frac{C_{2m}^{(|I|-I(0))/2}t^{\frac{m}{2}I(>0)+mI(0)+I(<0)}}{\sqrt{|I|!}},
\end{equation}
for some constant $C_{2m}>0$. Similarly, fix $N\in \N$ such that Condition \ref{eqn11} is satisfied. Then, for each
$m\leq 2N$ and $t\in [0,\delta],$ it also holds that
\begin{equation}
\E_0[|\e_n(t)|^{m}]\leq \frac{C_{2N}^{m(n+1)/2N}t^{m(n+1)/2N}}{((n+1)!)^{m/2N}},
\end{equation}
for some constant $C_{2N}>0$. 
\end{corollary}
\noindent \textbf{Proof.} The proof follows the same steps as that of Corollary \ref{corr1}.

\vspace{0.3cm}
Consistent with our findings in Subsection~\ref{proexp}, we conclude that Corollary \ref{cor1new} implies that the error in the general (i.e., for any $n$) process expansion in Eq.~(\ref{process2}) satisfies
\begin{equation*}
\e_n(t) = O_p(t^{n+1}),
\end{equation*}
for $N=1/2$.

\subsubsection{\enquote{Regular} moments}

We begin with the counterpart of Lemma \ref{lem3}. The lemma provides a representation of the drift of a process $f(\langle c,\overline \Z_t\rangle)$ (in isolation as well as multiplied by an arbitrary process $\langle d,\overline \Z_t\rangle$) in terms of a linear combination of the \textit{It\^o signature} elements. 

\begin{lemma}\label{lem3new}
For each $f\in C^2_p(\R),$ vectors $c, d,$ and $m$ large enough, it holds that
\begin{align*}
f\big(\langle c,\overline\Z_t\rangle\big)\langle d,\overline\Z_t\rangle
&=\int_0^t\bigg( \sum_{k=-e}^2\lambda_{-k}\int f^{(k^+)}\Big(\langle \Gcal_{k,\xi}^c(c),\overline\Z_s\rangle\Big)\langle\Gcal_{k,\xi}^d(c,d),\overline\Z_s\rangle \nu_{-k}(d\xi)
\bigg)ds\\
&\qquad+\text{martingale},
\end{align*}
where $\lambda_0=\lambda_{1}=\lambda_{2}=1$, $\nu_0=\nu_{1}=\nu_{2}$ are probability measures, $\Gcal_{k,\xi}^c(c)=c$ for $k\geq 0$ and $\Gcal_{k,\xi}^c(c)=\Jcal_{-k,\xi}(c)$ for $k<0$, and
$$
\Gcal_{k,\xi}^d(c,d):= \begin{cases}
 \Jcal_{-k,\xi}(d)&\text{if }k<0,\\
\sum_{H\in \Ical_n}d_HH'1_{\{h_{|H|}=0\}}-\sum_{j=1}^e\lambda_j d_HH&\text{if }k=0,\\
\sum_{I,H\in \Ical_n}c_Id_H I'\overline \star \Big(H1_{\{i_{|I|}=0\}}+ {\rho_{i_{|I|},h_{|H|}}}   H'\Big) &\text{if }k=1,\\
\sum_{I,J,H\in \Ical_n}c_Ic_Jd_H\frac {\rho_{i_{|I|},j_{|J|}}} 2I'\overline \star J'\overline \star H&\text{if }k=2,
\end{cases}$$
with $\Jcal_{j,\xi}(c)=\sum_{I\in \Ical_n}c_I\Big(I+ \xi^{\Lcal_2(i_{|I|})}1_{\{j=\Lcal_1(i_{|I|})\}}I'\Big)$.
\end{lemma}

\noindent \textbf{Proof.} See Appendix \ref{Appendix2}.

\vspace{0.3cm}
\noindent We may now expand the expectation of a linear combination of the \textit{It\^o signature} elements.

\begin{theorem}\label{thm1new}
For each $f\in C_p^{2N}(\R)$ with $N\in \N,$ and a vector $c$, it holds that
\begin{align*}
\E_0[f(\langle c,\overline \Z_t\rangle)]
&=f(c_\emptyset)
+\sum_{n=1}^N
\frac 1 {n!}\bigg(\sum_{k_1,\ldots,k_n=-e}^2\lambda_{k_1,\ldots,k_n}\int\cdots\int f^{(k_1^++\ldots+k_n^+)}\left(\langle\emptyset,\Gcal_{k_1,\ldots,k_n,\xi}^c(c)\rangle\right)\\
&\qquad \times \langle\emptyset,\Gcal_{k_1,\ldots,k_n,\xi}^d(c,\emptyset)\rangle\nu_{k_1,\ldots,k_n}(d\xi)\bigg)t^n
+o(t^N),
\end{align*}
where
$\Gcal^c_{k_1,\ldots,k_n,\xi} = \Gcal^c_{k_n,\xi} \circ ~\Gcal^c_{k_1,\ldots,k_{n-1},\xi}, ~ \Gcal^d_{k_1,\ldots,k_n,\xi}=\Gcal^d_{k_n,\xi}\circ ~ (\Gcal^c_{k_1,\ldots,k_{n-1},\xi},\Gcal^d_{k_1,\ldots,k_{n-1},\xi}),
$
$\lambda_{k_1,\ldots,k_n}=\lambda_{-k_1}\cdots\lambda_{-k_n}$, and 
$\nu_{k_1,\ldots,k_n}(d\xi)=\nu_{-k_1}(d\xi_1)\cdots \nu_{-k_n}(d\xi_n)$. 

\end{theorem}
\begin{proof}
The result follows from Lemma~\ref{lem3new} by using the same method of proofs as for Theorem~\ref{thm1}.
\end{proof}

\noindent Next, we turn to a suitably-differentiable function of the discontinuous process $X_t$ and characterize the distance between the conditional expectation of the function of the process and the conditional expectation of the same function applied to the process expansion. 

\begin{lemma}\label{lem4fnew}
Consider an $n+1$-times $\overline Z$-differentiable process $(X_t)_{t\in[0,T]}$ with expansion
$$
X_t=\langle c,\overline \Z_t\rangle+\e_n(t),
$$
satisfying Condition \ref{eqn11}. Then, for each $f\in C^1(\R)$ such that $\sup|f'|<\infty,$ it holds that
$$\E_0[f(X_t)]=\E_0[f(\langle c,\overline \Z_t\rangle)]+o(t^{n/2}).$$
\end{lemma}
\begin{proof}
The proof follows the same logic as that of Lemma~\ref{lem4f}.
 \end{proof}
 
\vspace{0.5cm}
 
\noindent Finally, we expand $\E_0[f(X_t)]$ by invoking the results above. 

\vspace{0.5cm}
 
\begin{theorem} \textbf{(Expanding \enquote{regular} moments.)}
Consider an $n+1$-times $\overline Z$-differentiable process $(X_t)_{t\in[0,T]}$ with expansion
$$
X_t=\langle c,\overline \Z_t\rangle+\e_n(t),
$$
satisfying Condition \ref{eqn11}. Then, for each $f\in C_p^{n+1}(\R)$ with $\sup|f'|<\infty$, it holds that
\begin{align*}
\E_0[f(X_t)]
&=f(X_0)
+\sum_{\ell=1}^{\lceil n/2\rceil}
\frac 1 {\ell!}\bigg(\sum_{k_1,\ldots,k_\ell=-e}^2\lambda_{k_1}\cdots\lambda_{k_\ell}\int\cdots\int f^{(k_1^++\ldots+k_\ell^+)}(\langle\emptyset,\Gcal_{k_1,\ldots,k_\ell,\xi}^c(c)\rangle)\\
&\times \langle\emptyset,\Gcal_{k_1,\ldots,k_\ell,\xi}^d(c,\emptyset)\rangle\nu_{k_1,\ldots,k_\ell}(d\xi)\bigg)t^\ell
+o(t^{n/2}).
\end{align*}
\end{theorem}
\begin{proof}
The claim follows from Theorem~\ref{thm1new} and Lemma~\ref{lem4fnew}, given $c_\emptyset=X_0$.
\end{proof}

\noindent We note that, through the representation provided in Remark~\ref{remcomp}, the coefficients of the expansion of $\E_0[f(X_t)]$ may be computed explicitly by means of an algorithm. An expansion of $\E_0[X_t^k]$ readily follows as an example. We begin with a lemma. 
 
\begin{lemma}\label{lem4new}
Consider an $n+1$-times $\overline Z$-differentiable process $(X_t)_{t\in[0,T]}$ with expansion
$$
X_t=\langle c,\overline \Z_t\rangle+\e_n(t),
$$
and fix $N\in \N$ such that Condition \ref{eqn11} is satisfied. Then, for each $k<2N$ it holds
$$\E_0[X_t^k]=\E_0[\langle c,\overline \Z_t\rangle^k]+o(t^{n/2}).$$
\end{lemma}
\begin{proof}
The proof follows the same logic as that of Lemma~\ref{lem4}.
 \end{proof}
 
\begin{example}[\textbf{$\mathbf{k^{th}}$ moment}]
Consider an $n+1$-times $\overline Z$-differentiable process $(X_t)_{t\in[0,T]}$ with expansion
$$
X_t=\langle c,\overline\Z_t\rangle+\e_n(t),
$$
satisfying Condition~\ref{eqn11} for some $N\in \N$. Then, for each $k<2N,$ it holds that 
\begin{align*}
\E_0[X_t^k]
&=X_0^k
+\sum_{\ell=1}^{\lceil n/2\rceil}
\frac 1 {\ell!}\bigg(\sum_{k_1,\ldots,k_\ell=-e}^2\lambda_{k_1,\ldots,k_\ell}\int\cdots\int \frac {k!}{(k-(k_1^++\ldots+k_\ell^+))!} (\langle\emptyset,\Gcal_{k_1,\ldots,k_\ell,\xi}^c(c)\rangle)^{k-(k_1^++\ldots+k_\ell^+)}\\
& \times \langle\emptyset,\Gcal_{k_1,\ldots,k_\ell,\xi}^d(c,\emptyset)\rangle\nu_{k_1,\ldots,k_\ell}(d\xi)\bigg)t^\ell
+o(t^{n/2}).
\end{align*}
\end{example}

\begin{proof}
The proof uses Theorem~\ref{thm1new} and Lemma~\ref{lem4new} as in the proof of Example \ref{exam}.
 \end{proof}
 
\subsubsection{\enquote{Irregular} moments}

In order to present the main result in this subsection, we begin with the analogue of Lemma \ref{lem2}. The lemma provides Fourier-like transforms of the \textit{It\^o signature} components with respect to two sources of randomness in the level process, the Brownian motion $W^1$ (like in Lemma \ref{lem2}) and the compound Poisson process $N^1$ (for which Lemma \ref{lem2} did not account).  

 \begin{lemma}\label{lem2J}
For each $I\in\Ical^k,$ it holds that
\begin{align*}
    &\E_0[\langle I,\Z_t\rangle\exp(iuc_1W_t^1+iuc_{-1}N_t^1)]\\
    &=\frac {t^k}{k! }\exp\Big(\frac{-(c_1u)^2t}2+\lambda_1t\int (\exp(iuc_{-1}\xi)-1) \nu_1(d\xi)\Big)\\
&
\times\begin{cases}
(iuc_1)^{I(1)}
(\lambda_1\int(\exp(iuc_{-1}\xi)\xi)\nu_1(d\xi))^{I(-1)}(\prod^e_{j=2}(\lambda_j\int \xi\nu_j(d\xi))^{I(-j)} )& \text{if }I=\{-e,\cdots,1\}^k,\\
0 &\text{else},
\end{cases}
\end{align*}
where we recall (c.f. Definition \ref{def2}) that $I(j)$ is the number of $j$s in $I$. 
\end{lemma}

\begin{proof}
See Appendix \ref{Appendix2}.
 \end{proof}

Generalizing now to a situation in which the level process contains both an idiosyncratic discontinuity $N^1$ and a discontinuity $N^2$ potentially correlated with discontinuities in the volatility process is straightforward. The corresponding result is contained in Lemma \ref{lem2J2}.  

 \begin{lemma}\label{lem2J2}
For each $I\in\Ical^k,$ it holds that
\begin{align*}
    &\E_0[\langle I,\Z_t\rangle\exp(iuc_1W_t^1+iuc_{-1}N_t^1+iuc_{-2}N_t^2)]\\
    &=\frac {t^k}{k! }\exp\Big(\frac{-(uc_1)^2t}2+\lambda_1t\int (\exp(iuc_{-1}\xi)-1) \nu_1(d\xi)+\lambda_2t\int (\exp(iuc_{-2}\xi)-1) \nu_2(d\xi)\Big)\\
&
\times\begin{cases}
(iuc_1)^{I(1)}
(\prod_{j=1}^2(\lambda_j\int(\exp(iuc_{-j}\xi)\xi)\nu_j(d\xi))^{I(-j)})(\prod_{j=3}^e(\lambda_j\int \xi\nu_j(d\xi))^{I(-j)} )& \text{if }I=\{-e,\cdots,1\}^k,\\
0 &\text{else},
\end{cases}
\end{align*}
where we recall (c.f. Definition \ref{def2}) that $I(j)$ is the number of $j$s in $I$. 
\end{lemma}

\begin{proof}
The proof is identical to that of Lemma \ref{lem2J}.
\end{proof}
As a direct consequence of Lemma \ref{lem2J2}, we obtain the following corollary. The corollary provides the asymptotic order of the Fourier transform of a generic \textit{It\^o signature} component and will be used repeatedly in our proofs. 
\begin{corollary}\label{corollario} For each $I\in \Ical_k,$ it holds
    $$\E_0\Big[\langle I,\Z_t\rangle\exp\Big(\frac{iu}{\sqrt t}(c_1W_t^1+iuc_{-1}N_t^1+iuc_{-2}N_t^2)\Big)\Big]=O\left(\frac{t^{k-I(1)/2}}{k!}\right).$$
\end{corollary}

\begin{proof}
Immediate after noticing that the characteristic exponent on the left hand side of the expression in Lemma \ref{lem2J2} has been divided by $\sqrt{t}$ and this change only affects the right hand side through the number of times ($I(1)$) the Brownian motion $W^1$ is iterated in the component $\langle I,\Z_t\rangle$.
\end{proof}
We are, once more, interested in the characteristic function of \textit{standardized} increments of the process. We begin with the (first two layers of the) assumed dynamics. Write 
\begin{eqnarray}\label{eqnX}
X_t&=&\int_0^t c_0(s) ds +\int_0^t c_1(s) dW_s^1+\int_0^t{c_{-1}}dN^1_s+\int_0^t{c_{-2}}dN^2_s\\
c_1(t)&=&c_1(0)+\int_0^t c_{01}(s) ds +\int_0^t c_{11}(s) dW_s^1+\int_0^t c_{21}(s) dW_s^2+\int_0^t c_{-2 1}(s) dN^2_s+\int_0^t c_{-31}(s)dN^3_s \notag \\
c_0(t)&=&\int_0^t c_{00}(s) ds +\int_0^t c_{10}(s) dW_s^1+\int_0^t c_{20}(s) dW_s^2+\int_0^t c_{30}(s) dW_s^3, \notag
\end{eqnarray}
where $c_{-21}$ and $c_{-31}$ are $( W,t)$-differentiable and 0 in 0. The continuous portion of the process satisfies Condition \ref{process} and is, therefore, consistent with the specification in Section~\ref{irregular}. To the process in Section~\ref{irregular}, however, we append idiosyncratic discontinuities in levels ($X$) and volatility ($c_1$), denoted by $N^1$ and $N^3$, as well as a joint discontinuity, denoted by $N^2$. The latter leads to a source of correlation between the level process and the volatility process which is distinct from the correlation induced by the common Brownian motion $W^1.$ We note that the assumed process is a \textit{nonparametric} version of specifications typically assumed in continuous-time finance, one in which the coefficients are unrestricted rather than being specified as parametric functions of the assumed state.

We may now state the central result in this subsection. The result is the analogue of Theorem~\ref{prop2}. Like in Theorem  \ref{prop2}, we could expand up to a generic order $t^n$. As earlier, the expansion to any order may be implemented through an algorithm. Here, without loss of generality, we stop the expansion to the second order in $\sqrt{t}.$ The second-order expansion will be employed in the subsequent corollary in order to derive the analogue (with price and volatility discontinuities) of the result in Corollary \ref{BR}.

\begin{theorem}\label{prop2J}
\textbf{(Expanding the characteristic function of the standardized process.)}
    Suppose that the process follows the dynamics described in Eq. \eqref{eqnX} and is a 4 times $(Z,t)$-differentiable process satisfying Condition~\ref{eqn11}. Given a vector $c$ such that $X_t=\langle c, \Z_t\rangle+\e_3(t)$ and given $m$ large enough, we obtain the following expansion:
\begin{align*}
\E_0\Big[\exp\Big(iu\frac{X_t}{c_1\sqrt t}\Big)\Big]
&= \E_0\bigg[ \exp\Big(iu\frac{c_1W_t^1+ c_{-1}N_t^1+c_{-2}N_t^2}{c_1 \sqrt t}\Big)\bigg]\\
&+
 \frac{iu}{ t^{1/2} }\sum_{I\in \Ical^{3}}
 \overline c_{I}\E_0\bigg[\langle I,\Z_t\rangle \exp\Big(iu\frac{c_1W_t^1+ c_{-1}N_t^1+c_{-2}N_t^2}{c_1 \sqrt t}\Big)\bigg]\\
 &+
 \frac{(iu)^2}{2 t }\sum_{I,J\in \Ical^{3}}
 \overline c_{I}\overline c_{J}\E_0\bigg[\langle I\overline\star J,\overline \Z_t\rangle \exp\Big(iu\frac{c_1W_t^1+ c_{-1}N_t^1+c_{-2}N_t^2}{c_1 \sqrt t}\Big)\bigg]+o(t),
\end{align*}
where $\overline c_1=\overline c_{-1}=\overline c_{-2}=c_0=0$ and $\overline c_I=c_I/c_1$ for each $I\notin\{-2,-1,0,1\}$.
\end{theorem}

\begin{proof}
See Appendix \ref{Appendix2}.
\end{proof}

\begin{corollary}\label{corr}\textbf{(A second-order (in $\sqrt{t}$) expansion.)}
Suppose that the process follows the dynamics described in Eq. \eqref{eqnX} and is a 4 times $(Z,t)$-differentiable process. Given a vector $c$ such that $X_t=\langle c, \Z_t\rangle+\e_3(t)$, we obtain the following expansion:
\begin{align*}
&\E_0\Big[\exp\Big(iu\frac{X_t-c_0t}{c_1\sqrt t}\Big)\Big]\\
&=\phi(u,c,t)+ t e^{-\frac{u^2}2}\Big(\lambda_1\int (e^{iu\frac{c_{-1}}{c_1\sqrt t }\xi}-1) \nu_1(d\xi)+\lambda_2\int (e^{iu\frac{c_{-2}}{c_1\sqrt t}\xi}-1) \nu_2(d\xi)
\Big)+o(t ),
\end{align*}
where $\phi(u,c,t)$ is the characteristic function expansion reported in Corollary~\ref{BR}.
\end{corollary}

\begin{proof}
See Appendix \ref{Appendix2}.
\end{proof}

The presence of level discontinuities adds two terms to the characteristic function expansion of the continuous portion of the process, i.e. $\phi(u,c,t),$ in Corollary~\ref{BR}. The role played by these two terms, jointly collected in 
\begin{equation}\label{term}
\lambda_1\int (e^{iu\frac{c_{-1}}{c_1\sqrt t }\xi}-1) \nu_1(d\xi)+\lambda_2\int (e^{iu\frac{c_{-2}}{c_1\sqrt t}\xi}-1)\nu_2(d\xi),
\end{equation}
is intuitive. Eq.~(\ref{term}) is the sum of the first-order expansions of the characteristic functions of the level discontinuities with intensities ($\lambda_1$ and $\lambda_2$) and jump measures ($\nu_1$ and $\nu_2$) frozen at 0.

The Corollary provides an alternative proof of the second-order expansion in Corollary 7 of \cite{BR:17}, one that is obtained, here, using the properties of the \textit{It\^o signature}. The slightly different look of the two expressions is an implication of the assumption $c_{-21} = c_{-31}=0,$ an assumption which eliminates the role of volatility discontinuities at zero. The assumption is only meant to aid the reader by simplifying the expression in Theorem~\ref{term} and clarifying its logic, as a result. Consistent with \cite{BR:17}, the idiosyncratic and joint discontinuities in volatility would, in general, play a role at time 0 since they would be of order $t.$ Their order is an implication of Corollary~\ref{corollario}. Notice, in fact, that the (standardized) Fourier transform of the generic \textit{It\^o signature} component $$ \frac{iu}{ t^{1/2} }\overline c_{I}\E_0\bigg[\langle I,\Z_t\rangle \exp\Big(iu\frac{c_1W_t^1+ c_{-1}N_t^1+c_{-2}N_t^2}{c_1 \sqrt t}\Big)\bigg]$$
in Theorem~\ref{prop2J} is, because of Corollary \ref{corollario}, of order $t$ either when $I=(-3,1)$ or when $I=(-2,1).$

\section{Further discussion and conclusions}\label{Conclusions}
An active area of research in (high-frequency) inference for continuous-time processes has recently used local expansions to estimate various aspects of the process of interest (e.g., \citealp{BR:17}, \citealp{To:21}, and \citealp{chong2024volatility}). 

Two typical features of these local expansions is that they apply to a rather specific object, namely the conditional characteristic function of the process, and they are low order. Emphasis on the conditional characteristic function is, of course, justified by its one-to-one mapping with the process' density and, therefore, its usefulness, e.g., in efficient inferential procedures. While generally appropriate for the specific task at hand, emphasis on low-order approximations is also, undoubtedly, a result of the complexity of these expansions and the corresponding need to track an increasing number of terms for every additional order. As an example, the second-order expansion of the characteristic function of the standardized process in Corollary \ref{BR} (which re-derives - using the methods proposed in this article - central results in \citealp{BR:17}, and \citealp{To:21}) has 1 term to the first order (in $\sqrt{t}$) and 8 to the second order (in $t$). As shown in Corollary \ref{BR1}, which contains a novel third-order expansion of the same object, the third order (in $t^{3/2}$) contains 19 terms, thereby resulting in nearly exponential growth of the number of terms for each additional order in the expansion. Deriving the 19 third-order terms using existing methods of proof is \enquote{expensive}. In addition, existing methods of proof would not automate the expansions to higher orders, thereby requiring \enquote{increasingly expensive} calculations. Yet, interest may be in the dynamics of deep characteristics which only higher-order expansions would reveal.         
  
Against this backdrop, this article provides local expansions of semimartingales and their moments using the properties of the \textit{It\^o signature}. The characteristic function of the standardized process is a sub-case of our set of results, which are for conditional moments of functions with well-defined derivatives at zero as well as for conditional moments of functions with unbounded derivatives at zero. Low-order expansions are also a sub-case of our set of results. Making use of fundamental properties of the \textit{It\^o signature}, we show how expansions of generic local moments to any order can be derived and implemented through algorithms. The end result is arbitrarily-accurate representations of local moments for general functions of processes with and without compound Poisson discontinuities. 

Going forward, our methods offer an alternative way to conduct asymptotics in the context of these expansions. Not only can we assume time to become increasingly small given a specific order of the expansion, we may also work with asymptotic designs in which the order of the expansion enlarges asymptotically. As emphasized, the latter approach is expected to permit a deeper dive into layers of the process and, therefore, identification of quantities other than $c_1$ (spot volatility), $c^2_{11}+c^2_{12}$ (spot variance of volatility) and $c_{11}$ (leverage), the exclusive focus of the current high-frequency literature. It could also permit superior identification of these same quantities, when time is not especially short.

The proposed methods could also be used to automate moment expansions of suitable modifications of the stochastic process assumed in the current article. We provide three examples. First, while our emphasis is on economically-meaningful (large) discontinuities, jumps of infinity activity (and infinity variation) may be added. Second, fractional Brownian motion may be one of the shocks affecting the process characteristics (e.g., $c_1$, spot volatility), thereby permitting forms of  \enquote{roughness} (\citealp{gatheral2018volatility}). The recent characteristic function expansion in \citet{chong2022short} makes progress along both dimensions while remaining low order. Finally, we may automate expansions of (moments of) multivariate processes beyond the bivariate (for levels and volatility) second-order characteristic function expansion in \cite{BR:17}. These lines of inquiry are better left for future work.  

\appendix     

\section{Appendix: Proofs in the continuous case} \label{Appendix1}

\begin{small}

\begin{proof}[Proof of Lemma \ref{prop}]
We proceed by induction. Since $\langle I,\widehat\W_t\rangle=1$ for $I=\emptyset,$ the result is clear if $|I|=0$ or $|J|=0$. Assume, now, that the claim holds for each $I,J$ such that $|I|+|J|\leq n-1$. Thus, for each $I,J$ such that $|I|+|J|\leq n,$ 
by It\^o's formula and since $[\widehat W^i,\widehat W^j]_t=\rho_{ij} t=\rho_{ij}\widehat W_t^0,$ we have that
\begin{align*}
\langle I,\widehat\W_t\rangle\langle J,\widehat\W_t\rangle
&=\bigg(\int_0^t\langle I',\widehat\W_s\rangle d\widehat W_s^{i_{|I|}}\bigg)\bigg(\int_0^t\langle J',\widehat\W_s\rangle d\widehat W_s^{j_{|J|}}\bigg)\\
&=\int_0^t\langle J,\widehat\W_s\rangle\langle I',\widehat\W_s\rangle d\widehat W_s^{i_{|I|}}+\int_0^t\langle I,\widehat\W_s\rangle\langle J',\widehat\W_s\rangle d\widehat W_s^{j_{|J|}} \\
&+\rho_{i_{|I|}j_{|J|}}\int_0^t\langle J',\widehat\W_s\rangle\langle I',\widehat\W_s\rangle d\widehat W^0_s.
\end{align*}
By induction, the claim follows.
\end{proof}

\begin{proof}[Proof of Lemma \ref{lemm9}]
Observe that for $H_t$ such that $t\mapsto \E_0[H_t^{2K}]$ is integrable on $[0,\delta]$ we have that
\begin{align*}
\E_0\bigg[\bigg(\int_0^tH_sds\bigg)^{2K}\bigg]
&\leq t^{2K-1}\E_0\bigg[\int_0^tH_s^{2K}ds\bigg]
= t^{2K-1}\int_0^t\E_0[H_s^{2K}]ds,\\
\E_0\bigg[\bigg(\int_0^tH_sdW^i_s\bigg)^{2K}\bigg]
&\leq C_{2K}\E_0\bigg[\bigg(\int_0^tH_s^2ds\bigg)^K\bigg]
\leq C_{2K}t^{K-1} \int_0^t\E_0[H_s^{2K}]ds,
\end{align*}
for some constant $C_{2K}>0$ (which depends on $K$) changing from place to place.\footnote{We use the same convention for all constants in these proofs.} Here, Jensen's inequality has been used in the first and in the third inequality and BDG inequality has been used in the second. We, thus, obtain that
$$\E_0\bigg[\bigg(\int_0^tH_sd\widehat W_s^i\bigg)^{2K}\bigg]
\leq C_{2K}^{1-1_{\{i=0\}}}t^{K(1+1_{\{i=0\}})-1}\int_0^t\E_0[H_s^{2K}]ds.$$
Recursively applying this argument, we get - using Lemma \ref{lem_moments} - that 
\begin{align*}
\E_0\bigg[\bigg(\int_0^t\int_0^{t_{n}}\cdots\int_0^{t_{2}} H_{t_{1}} d\widehat W_{t_{1}}^{i_{1}}\cdots d\widehat W_{t_{|I|}}^{i_{|I|}}\bigg)^{2K}\bigg]
&\leq C_{2K}^{|I|-I(0)}t^{(K-1)|I|+KI(0)}\frac{t^{|I|}}{|I|!}\sup_{t\in[0,\delta]}\E_0[H_{t}^{2K}],
\end{align*}
and the claim follows.
\end{proof}

\begin{proof}[Proof of the Corollary \ref{corr1}]
We observe that, by Jensen's inequality,
$$\E_0[\langle I,\widehat\W_t\rangle^{m}]\leq \E_0[\langle I,\widehat\W_t\rangle^{2m}]^{1/2}\qquad
\text{and} 
\qquad\E_0[|\e_n(t)|^{m}]\leq \E_0[|\e_n(t)|^{2N}]^{m/2N},$$
since $m\leq 2N.$ The claim follows directly from Lemma~\ref{lemm9} by choosing $H_t = 1$ and, because of Eq.~(\ref{error}), $H_t = \sum_{|I|=n+1}c_I(t)$, respectively.
\end{proof}

\begin{proof}[Proof of Lemma \ref{lem3}]
Given Remark \ref{bbbrem}, we have that
$$\langle I,\widehat\W_t\rangle=\int_0^t\langle I',\widehat \W_s\rangle1_{\{i_{|I|}=0\}} ds+\text{martingale}.$$
By {Lemma~\ref{prop}}, this representation implies that
$$[\langle I,\widehat\W\rangle,\langle J,\widehat\W\rangle]_t=\int_0^t\langle I'\star J',\widehat \W_s\rangle\rho_{i_{|I|},j_{|J|}}ds.$$
By It\^o's lemma, we then have
\begin{eqnarray*}
df\Big(\langle c,\widehat \W_t\rangle\Big) 
&=&f'\Big(\langle c,\widehat \W_t\rangle\Big)
\sum_{I\in \Ical_n}c_I1_{\{i_{|I|=0}\}}\langle I',\widehat \W_t\rangle dt\\
&+&f'\Big(\langle c,\widehat \W_t\rangle\Big)
\sum_{I\in \Ical_n}c_I1_{\{i_{|I|>0}\}}\langle I',\widehat \W_t\rangle d W^{i_{|I|}}_t\\
&+&\frac 1 2 f''\Big(\langle c,\widehat \W_t\rangle\Big)
\sum_{I,J\in \Ical_n}c_Ic_J{\rho_{i_{|I|},j_{|J|}}}\langle I'\star J',\widehat \W_t\rangle dt.
\end{eqnarray*}
The claim follows, again, from It\^o's lemma, since 
\begin{eqnarray*}
f\Big(\langle c,\widehat \W_t\rangle\Big)\langle d,\widehat \W_t\rangle &=& \int_0^t f\Big(\langle c,\widehat \W_s\rangle\Big)\sum_{H\in \Ical_n}d_H 1_{\{\red h_{|H|}=0\}} \langle H' , \widehat \W_s \rangle ds \\
&+& \int_0^t f'\Big(\langle c,\widehat \W_s\rangle\Big)\sum_{I,H\in \Ical_n}c_Id_H 1_{\{i_{|I|}=0\}} \langle I' \star H, \widehat \W_s \rangle ds \\
&+&  \int_0^t \frac 1 2 f''\Big(\langle c,\widehat \W_s\rangle\Big)
\sum_{I,J\in \Ical_n}c_Ic_Jd_H{\rho_{i_{|I|},j_{|J|}}}\langle I'\star J' \star H,\widehat \W_s\rangle ds \\
&+& \int_0^t f'\Big(\langle c,\widehat \W_s\rangle\Big)
\sum_{I,H\in \Ical_n}c_I d_H\rho_{\red i_{|I|},h_{|H|}}\langle I' \star H' ,\widehat \W_s\rangle ds \\
&+& \text{martingale}.
\end{eqnarray*}
Observe that, by the assumed growth condition, the last term is a true martingale and not just a local martingale.
\end{proof}

\begin{proof}[Proof of Theorem~\ref{thm1}]
By Lemma~\ref{lem3}, we have that
$$\partial_t\E_0\Big[f\Big(\langle c,\widehat \W_t\rangle\Big)\langle d,\widehat \W_t\rangle\Big]
=\sum_{k=0}^2\E_0[f^{(k)}(\langle c,\widehat \W_t\rangle)\langle\Gcal_{c,k}(d),\widehat \W_t\rangle].$$
Because Lemma~\ref{lem3} is written for $f\Big(\langle c,\widehat \W_t\rangle\Big)\langle d,\widehat \W_t\rangle$ rather than for $f\Big(\langle c,\widehat \W_t\rangle\Big),$ we can easily iterate the computation of derivatives. The $n^{th}$ derivative is, in fact, 
$$\partial_t^n\E_0\Big[f\Big(\langle c,\widehat \W_t\rangle\Big)\langle d,\widehat \W_t\rangle\Big]
=\sum_{k_1,\ldots,k_n=0}^2\E_0[f^{(k_1+\ldots+k_n)}(\langle c,\widehat \W_t\rangle)\langle\Gcal_{c,k_1,\ldots,k_n}(d),\widehat \W_t\rangle],$$
for each vector $d$. Recalling that $\langle \emptyset,\widehat \W_t\rangle = 1,$ the claim follows by an application of Taylor's approximation theorem to the map $t\mapsto \E_0[f\Big(\langle c,\widehat \W_t\rangle\Big)\langle \emptyset,\widehat \W_t\rangle]$ which is expanded to the order $N$ afforded by the differentiability of the function $f(.)$.
\end{proof}

\begin{proof}[Proof of Lemma \ref{lem4f}]
Observe that, because $X_t=\langle c,\widehat\W_t\rangle+\e_n(t),$ by the mean-value theorem, we may write
$$f(X_t)=f(\langle c,\widehat\W_t\rangle)+Y\e_n(t),$$
for some random variable $Y$ satisfying $|Y|\leq \textrm{sup}|f'| \leq C$.
Cauchy-Schwartz inequality and Corollary~\ref{corr1} now give
$$\E_0[|Y\e_n(t)|] \leq \E_0[Y^2]^{1/2}\E_0[\e_n(t)^2]^{1/2}\leq C\E_0[\e_n(t)^2]^{1/2}=o(t^{n/2})$$
and the claim follows.
 \end{proof}
 
 \begin{proof}[Proof of Lemma \ref{lem4}]
Observe that, since $X_t=\langle c,\widehat\W_t\rangle+\e_n(t),$ the binomial formula gives 
$$X_t^k=\langle c,\widehat\W_t\rangle^k+\sum_{\ell=1}^k\binom k \ell \langle c,\widehat\W_t\rangle^{k-\ell}\e_n(t)^{\ell}.$$
By Holder's inequality, we know that 
$$\E_0[\langle I,\widehat\W_t\rangle^{k-\ell}\e_n(t)^{\ell}]\leq \E_0[\langle I,\widehat\W_t\rangle^{2m}]^{(k-\ell)/2m}\E_0[\e_n(t)^{2N}]^{\ell/2N}$$
for $m$ such that $\frac {k-\ell}{2m}+\frac \ell {2N}=1$. Since, by Corollary~\ref{corr1} and Remark \ref{rem1}, the first term is bounded on $[0,\delta]$ and the second term is bounded by $C t^{\ell(n+1)/2}$ (and, thus, by $C t^{(n+1)/2}$), the claim follows.
 \end{proof}
 
 \begin{proof}[Proof of Lemma \ref{lem2}]
Observe that, for $j\in\{1,\ldots,d\},$ by It\^o's Lemma, we have
\begin{eqnarray} \label{exx}
d\exp(iuW_t^1+u^2t/2)&=&iu\exp(iuW_t^1+u^2t/2)dW_t^1 - \frac{u^2}{2}\exp(iuW_t^1+u^2t/2)dt + \frac{u^2}{2}\exp(iuW_t^1+u^2t/2)dt \notag\\
&=&iu\exp(iuW_t^1+u^2t/2)dW_t^1.
\end{eqnarray}
Thus, for each $j>0$,
\begin{eqnarray}\label{exp}
d\langle Jj,\widehat\W_t\rangle\exp(iuW_t^1+u^2t/2)&=&\exp(iuW_t^1+u^2t/2)\langle J,\widehat\W_t\rangle dW_t^{j}+\langle Jj,\widehat\W_t\rangle iu\exp(iuW_t^1+u^2t/2) dW_t^1 \notag \\
&&+iu\exp(iuW_t^1+u^2t/2)\langle J,\widehat\W_t\rangle 1_{\{j=1\}} dt
\end{eqnarray}
and
\begin{eqnarray*}
&&d\langle J0,\widehat\W_t\rangle\exp(iuW_t^1+u^2t/2)=\exp(iuW_t^1+u^2t/2)\langle J,\widehat\W_t\rangle dt
+\langle J0,\widehat\W_t\rangle iu\exp(iuW_t^1+u^2t/2) dW_t^1.
\end{eqnarray*}
The last two expressions imply that
$$\E_0[\langle Jj,\widehat\W_t\rangle\exp(iuW_t^1+u^2t/2)]=(1_{\{j=0\}}+iu{1_{\{j=1\}}})\E_0\left[\int_0^t\langle J,\widehat\W_s\rangle\exp(iuW_s^1+u^2s/2) ds
\right].$$
Proceeding by recursion for $I\in \{0,1\}^k,$ we obtain
\begin{align*}
\E_0[\langle I,\widehat\W_t\rangle\exp(iuW_t^1+u^2t/2)]
&=
(iu)^{I(1)}\int_0^t\int_0^{t_1}\cdots \int_0^{t_{k-1}}\E_0[\exp(iuW_{t_k}^1+u^2{t_k}/2)] d {t_k}\cdots d {t_1}\\
&=(iu)^{I(1)}\frac {t^k}{k! }.
\end{align*}
Finally, Eq.~(\ref{exp}) clarifies that the expectation vanishes for each $I\notin \{0,1\}^k$.
\end{proof}

\begin{proof}[Proof of Theorem~\ref{prop2}]
As $X_t=\langle  c,\widehat \W_t\rangle+\e_{m+1}(t)$, an application of the mean-value theorem\footnote{Recall that the mean-value theorem does not hold, in general, for complex-valued functions.} for the real and the imaginary parts of $f(x)=\exp({iux})$ yields
\begin{align*}
f\Big(\frac{X_t}{c_1\sqrt t}\Big)&= f\Big(\frac{\langle  c,\widehat \W_t\rangle}{c_1\sqrt t}\Big)+Z\frac{\e_{m+1}(t)}{c_1\sqrt t},
\end{align*}
for some random variable $Z$ satisfying $|Z|\leq \textrm{sup}| \Re(f')|+\textrm{sup}| \Im(f')| \leq C$. By Condition~\ref{eqn11}, an application of Cauchy-Schwartz - as in the proof of Lemma~\ref{lem4f} - yields 
$$\E_0\Big[f\Big(\frac{X_t}{c_1\sqrt t}\Big)\Big]= \E_0\Big[f\Big(\frac{\langle  c,\widehat \W_t\rangle}{c_1\sqrt t}\Big)\Big]+o(t^{m/2}).$$
Next, note that $\langle  c,\widehat \W_t\rangle/c_1=\langle  \overline c,\widehat \W_t\rangle+W^1_t$ and, hence,
$\E_0\Big[\exp\Big(iu\frac{\langle  c,\widehat \W_t\rangle}{c_1\sqrt t}\Big)\Big]
=\E_0\Big[f\Big(\frac{W_t^1}{\sqrt t}\Big)f\Big(\frac{\langle  \bar c,\widehat \W_t\rangle}{\sqrt t}\Big)\Big]$. An application of Taylor's approximation theorem to the real and the imaginary parts of $f$, thus, yields
\begin{align*}
f\Big(\frac{\langle  \bar c,\widehat \W_t\rangle}{\sqrt t}\Big)&=\sum_{\ell=0}^m \frac{f^{(\ell)}(0)}{t^{\ell/2}\ell!}\langle  \bar c,\widehat \W_t\rangle^\ell+R^m_f(\langle  \bar c,\widehat \W_t\rangle),
\end{align*}
where
$$|R^m_f(y)|\leq\frac{\textrm{sup}|\Re(f^{(m+1)})|+\textrm{sup}|\Im(f^{(m+1)})|}{(m+1)!}\frac{|y|^{m+1}}{t^{(m+1)/2}}
\leq 2\frac{\textrm{sup}|f^{(m+1)}|}{(m+1)!}\frac{|y|^{m+1}}{t^{(m+1)/2}} ,$$
 with $\textrm{sup}|f^{(m+1)}|<C$. By Corollary~\ref{corr1} we get that
$$\E_0[|\langle  \bar c,\widehat \W_t\rangle|^{m+1}]\leq C\sum_{I\in \Ical^{m+1}} |\bar c_I|\E_0[|\langle I,\widehat \W_t\rangle|^{m+1}]\leq C\sum_{I\in \Ical^{m+1}} |\bar c_I|t^{(m+1)(|I|+I(0))/2},$$
for some potentially different constants $C$.
As $\bar c_\emptyset=\bar c_1=0$ by Condition~\ref{processcond}, this can be bounded by $C t^{m+1}$, implying in particular that
$$\E_0[|R^m_f(\langle  \bar c,\widehat \W_t\rangle)|]\leq
C\frac{\E_0[|\langle  \bar c,\widehat \W_t\rangle|^{m+1}]}{t^{(m+1)/2}}
=o(t^{m/2}).$$
The claim now follows from Lemma~\ref{prop} and the observation that $f^{(\ell)}(0)=(iu)^\ell$.
\end{proof}

\begin{proof}[Proof of Corollary \ref{BR}]
By Theorem~\ref{prop2}, for $m=2$, we have 
\begin{align*}
\E_0\Big[\exp\Big(iu\frac{X_t}{c_1\sqrt t}\Big)\Big]
&=\E_0\bigg[ \exp\Big(\frac{iu}{\sqrt t}{W_t^1}\Big)\bigg]
+ \frac{iu}{\sqrt t }\sum_{I\in \Ical^3}\overline c_{I}\E_0\bigg[\langle I,\widehat\W_t\rangle \exp\Big(\frac{iu}{\sqrt t}{W_t^1}\Big)\bigg]\\
&\qquad+\frac{(iu)^2}{2t }\sum_{I_1, I_2 \in \Ical^3}\overline c_{I_1} \overline c_{I_2}\E_0\bigg[\langle I_1\star I_2,\widehat\W_t\rangle \exp\Big(\frac{iu}{\sqrt t}{W_t^1}\Big)\bigg] +o(t ),
\end{align*}
Observe that, setting $u_t:=u/\sqrt t,$ by Lemma~\ref{lem2} we have that $\E_0[\exp(iu_tW_t^1)]= e^{-\frac{u^2} 2},$
\begin{align} \label{aaaeq}
\sum_{I\in \Ical^3}&\overline c_{I}\E_0\bigg[\langle I,\widehat\W_t\rangle \exp\Big(iu_t{W_t^1}\Big)\bigg] \notag \\
 &=
 e^{-\frac{u^2} 2}
\bigg[
(\overline c_{01}+\overline c_{10})(i u_t)\frac{t^{2}}2
+\overline c_{11}(i u_t)^2\frac{t^{2}}2
+\overline c_{111}(i u_t)^3\frac{t^{3}}6\bigg]+o(t^{3/2}),
\end{align}
and
\begin{align}\label{bbb}
\sum_{I_1, I_2 \in \red \Ical^3}&\overline c_{I_1} \overline c_{I_2}\E_0\bigg[\langle I_1\star I_2,\widehat\W_t\rangle \exp\Big(iu_t{W_t^1}\Big)\bigg] \notag \\
&= e^{-\frac{u^2} 2}
\bigg[
\overline c_{11}^2\left(\frac{t^2}2+6(i  u_t)^2\frac{t^3}6+6(i  u_t)^4\frac{t^4}{24}\right)
+\overline c_{21}^2\left(2(iu_t)^2\frac {t^3}{6}+\frac {t^2}{2}\right)\bigg]+o(t^2).
\end{align}
We note that, in Eq.~(\ref{aaaeq}), the terms with coefficients $\overline c_{0}$ and $\overline c_{1}$ do not appear since those coefficients are zero and the terms with coefficients $\overline c_{00}, \overline c_{000}, \overline c_{001}$ etc. are folded into the error term. Similar considerations apply to Eq.~(\ref{bbb}). Eq.~(\ref{bbb}) also derives from the properties of the $\star$ operator. For example, the first term, e.g., is due to the fact that $11\star11 = 6(1111) + 2(011)+2(101)+2(110)+00,$ given Lemma \ref{prop}. 

Combining Eq.~(\ref{aaaeq}) and Eq.~(\ref{bbb}),  we obtain
\begin{align*}
\E_0\Big[\exp\Big(iu\frac{X_t-c_0t}{c_1\sqrt t}\Big)\Big]
&=e^{-\frac{u^2} 2}\bigg(1
+{iu}{ }
[(\overline c_{01}+\overline c_{10})(iu)\frac{t}2
+\overline c_{11}(i u)^2\frac{\sqrt t}2
+\overline c_{111}(i u)^3\frac{ t}6]\\
&\qquad-\frac{u^2}{2  }
[\overline c_{11}^2(\frac{t}2+6(i u)^2\frac{t}6+6(i u)^4\frac{t}{24})
+\overline c_{21}^2(2(iu)^2\frac {t}{6}+\frac {t}{2})]\bigg)+o(t)\\
&=e^{-\frac{u^2} 2}\bigg(1
+[-\overline c_{11}\frac{iu^3}2]{\sqrt t}\\
&\qquad
+\frac 1 2 [-(\overline c_{01}+\overline c_{10})u^2
+\overline c_{11}^2(-\frac 1 2 u^2+u^4-\frac{1}{4}u^6)
+\overline c_{21}^2(\frac {1}{3}u^4-\frac {1}{2}u^2)
+\overline c_{111}\frac{u^4}{3}]t\bigg)\\
&\qquad+o(t),
\end{align*}
thereby proving the claim.
\end{proof}

\section{Appendix: Proofs in the discontinuous case} \label{Appendix2}

\begin{proof}[Proof of Lemma \ref{newprop}]
We proceed by induction. Since $\langle J,\overline \Z_t\rangle=1$ for $J=\emptyset$ the result is clear if $|I|=0$ or $|J|=0$. Assume then that the claim holds for each $I,J$ such that $|I|+|J|\leq n-1$. Then, for each $I,J$ such that $|I|{+}|J|\leq n,$ 
by definition of signature and the product formula, we have that
\begin{align*}
\langle I,\overline\Z_t\rangle\langle J,\overline\Z_t\rangle
&=\bigg(\int_0^t\langle I',\overline\Z_s\rangle d\overline Z_s^{i_{|I|}}\bigg)\bigg(\int_0^t\langle J',\overline \Z_s\rangle d\overline Z_s^{j_{|J|}}\bigg)\\
&=\int_0^t\langle I', \overline \Z_s\rangle\langle J, \overline \Z_s\rangle d\overline Z_s^{i_{|I|}}
+\int_0^t\langle J', \overline \Z_s\rangle\langle I, \overline \Z_s\rangle d\overline Z_s^{j_{|J|}}\\
&\qquad+ \int_0^t\langle I', \overline \Z_s\rangle\langle J', \overline \Z_s\rangle d[\overline Z^{i_{|I|}},\overline Z^{j_{|J|}}]_s^c
+\sum_{s\leq t}\langle I', \overline \Z_s\rangle\langle J', \overline \Z_s\rangle \Delta\overline Z^{i_{|I|}}_s\Delta \overline Z^{j_{|J|}}_s.
\end{align*}
Since
$$ \int_0^t\langle I', \overline \Z_s\rangle\langle J', \overline \Z_s\rangle d[\overline Z^i, \overline Z^j]_s^c
=\int_0^t\langle I', \overline \Z_s\rangle\langle J', \overline \Z_s\rangle 1_{\{i=j>0\}}ds
=1_{\{i=j>0\}}\int_0^t\langle I', \overline \Z_s\rangle\langle J', \overline \Z_s\rangle d\overline Z^0_s
$$
and, for each $i,j<0,$
\begin{align*}
\sum_{s\leq t}\langle I', \overline \Z_s\rangle\langle J', \overline \Z_s\rangle \Delta\overline Z^i_s\Delta \overline Z^j_s
&=
\sum_{s\leq t}\langle I', \overline \Z_s\rangle\langle J', \overline \Z_s\rangle (\Delta  N_s^{\Lcal_1(i)})^{\Lcal_2(i)+\Lcal_2(j)}1_{\{\Lcal_1(i)=\Lcal_1(j)\}}\\
&=
\int_0^t \langle I', \overline \Z_s\rangle\langle J', \overline \Z_s\rangle dN_s^{\Lcal_1(i),\Lcal_2(i)+\Lcal_2(j)}1_{\{\Lcal_1(i)=\Lcal_1(j)\}}\\
&=
\int_0^t \langle I', \overline \Z_s\rangle\langle J', \overline \Z_s\rangle d\overline Z_s^{i\circ j}1_{\{\Lcal_1(i)=\Lcal_1(j)\}},
\end{align*}
the claim follows.
\end{proof}

\begin{proof}[Proof of Lemma \ref{lem_momentsJ}]
For for each $i<0$, we have
$$\E_0\left[\int_0^t H_s d\overline Z_s^i\right]=\E_0\left[\int_0^t H_s \lambda_{\Lcal_1(i)}\int  \xi^{\Lcal_2(i)} \nu_{\Lcal_1(i)}(d\xi)ds\right]
=\E_0\left[\int_0^t H_sds\right]\lambda_{\Lcal_1(i)}\int \xi^{\Lcal_2(i)} \nu_{\Lcal_1(i)}(d\xi).$$ 
Iterating over $H_s$ yields the result.
\end{proof}

\begin{proof}[Proof of Lemma \ref{lem9new}]
For each $j\in \{-1,\ldots, -e\cdot m\}$, write $\overline \lambda_j=\lambda_{\Lcal_1(j)}$ and 
$$\xi^{\star} \overline \nu_j(d\xi)=\xi^{\Lcal_2(j)} \nu_{\Lcal_1(j)}(d\xi).$$
 Using $(a+b)^{2K}\leq 2^{2K-1}(a^{2K}+b^{2K})$, an application of It\^o's formula yields
\begin{align*}
\E_0\Big[\Big(\int_0^tH_rd\overline { Z}_r^j\Big)^{2K}\Big]
&=\int_0^t\E_0\Big[\overline \lambda_j\int\Big(\int_0^sH_rd\overline { Z}_r^j+H_s\xi^\star\Big)^{2K}-\Big(\int_0^sH_rd\overline { Z}_r^j\Big)^{2K}\Big] \overline \nu_j(d\xi) ds\\
&\leq C_{2K}\int_0^t\overline \lambda_j\int\E_0\Big[\Big(\int_0^sH_rd\overline {Z}_r^j\Big)^{2K}+(H_s\xi^\star)^{2K}\Big] \overline \nu_j(d\xi)ds\\
&=C_{2K}\overline \lambda_j\int_0^t\E_0\Big[\Big(\int_0^sH_rd\overline {Z}_r^j\Big)^{2K}\Big]ds+C_{2K}\int_0^t\E_0[H_s^{2K} ]ds\Big(\overline \lambda_j\int(\xi^\star)^{2K}\overline \nu_j(d\xi)\Big).
\end{align*}
By Gr\"onwall's inequality, for $t\leq 1,$ this implies
$$\E_0\Big[\Big(\int_0^tH_rd\overline { Z}_r^j\Big)^{2K}\Big]\leq C_{2K}\int_0^t\E_0[H_s^{2K} ]ds \exp(C_{2K}\overline \lambda_j t)\leq C_{2K}\int_0^t\E_0[H_s^{2K} ]ds,$$
for a potentially changing constant $C_{2K}$. Proceeding now as in the proof of Lemma~\ref{lemm9}, we may conclude that
$$\E_0\bigg[\bigg(\int_0^t\int_0^{t_{|I|}}\cdots\int_0^{t_{2}} H_{t_{1}} d\overline Z_{t_{1}}^{i_{1}}\cdots d\overline Z_{t_{|I|}}^{i_{|I|}}\bigg)^{2K}\bigg]\leq  \frac{C_{2K}^{|I|-I(0)}t^{KI(>0)+2KI(0)+I(<0)}}{|I|!}\sup_{t\in[0,\delta]}\E_0[H_{t}^{2K}].$$
\end{proof}

\begin{proof}[Proof of Lemma \ref{lem3new}]
By the definition of signature, we know that $\langle  I, \overline \Z_t \rangle$ is a semimartingale. Its drift component is given by $\int_0^t \langle I',\overline \Z_s\rangle1_{\{i_{|I|}=0\}}ds$ and its continuous quadratic co-variation with respect to a generic component $J$ is 
$[\langle  I, \overline \Z \rangle,\langle  J, \overline \Z \rangle]_t=\int_0^t\langle I'\ostar J',\overline \Z_s\rangle\rho_{i_{|I|},j_{|J|}}ds.$ We define the quantity $\gamma_j(\overline \Z_s,\xi)$ as 
$$\langle I, \gamma_j(\overline \Z_s,\xi)\rangle
=\langle I',\overline \Z_s\rangle \xi^{\Lcal_2(i_{|I|})}1_{\{\Lcal_1(i_{|I|})=j\}}.$$
By It\^o's formula's, the drift of the process $f\big(\langle c,\overline \Z_t\rangle\big)\langle d,\overline \Z_t\rangle$ can be written as follows:
\begin{align*}
&\int_0^t f'\big(\langle c,\overline \Z_s\rangle\big)\langle d,\overline \Z_s\rangle\sum_{I\in \Ical_n}c_I\langle I',\overline \Z_s\rangle1_{\{i_{|I|}=0\}}ds \\
&\qquad+\int_0^t f\big(\langle c,\overline \Z_s\rangle\big)\sum_{H\in \Ical_n}d_H\langle H',\overline \Z_s\rangle1_{\{h_{|H|}=0\}}ds\\
&\qquad+\frac 1 2 \int_0^t f''\big(\langle c,\overline \Z_s\rangle\big)\langle d,\overline \Z_s\rangle\sum_{I,J\in \Ical_n}c_Ic_J\langle I'\overline \star J',\overline \Z_s\rangle\overline \rho_{i_{|I|},j_{|J|}}ds\\
&\qquad+\int_0^tf'\big(\langle c,\overline \Z_s\rangle\big)\sum_{I,H\in \Ical_n}c_Id_H\langle I'\overline \star H',\overline \Z_s\rangle\overline \rho_{i_{|I|},h_{|H|}}ds\\
&\qquad+\sum_{j=1}^e \lambda_j\int_0^t \int f\big(\langle c,\overline \Z_s+\gamma_j(\overline \Z_s,\xi)\rangle\big)\langle d,\overline \Z_s+\gamma_j(\overline \Z_s,\xi)\rangle
-f\big(\langle c,\overline \Z_s\rangle\big)\langle d,\overline \Z_s\rangle
 \nu_j(d\xi) ds,
\end{align*}
where the last term is the jump compensation. In the statement of the theorem, the second piece of the compensation is the second piece of $\Gcal_{k,\xi}^d(c,d)$ with $k=0$. The claim now follows.
\end{proof}

 \begin{proof}[Proof of Lemma \ref{lem2J}]
For convenience, but without loss of generality, we set $c_1=1$. Define
$$A(u):=\lambda_1\int (\exp(iuc_{-1}\xi)-1) \nu_1(d\xi).$$ Denoting by $\mu_1$ the jump measure of $N^1,$ by It\^o's Lemma it holds that   
\begin{eqnarray}\label{exx3}
 &&d\exp(iuc_{-1}N_t^1-A(uc_{-1})t) \notag\\
 &=&\exp(iuc_{-1}N^1_t - A(uc_{-1})t)\int (\exp(iuc_{-1}\xi)-1)(\mu_1(dt,d\xi)-\lambda_1\nu_1(d\xi)dt),
\end{eqnarray}
which is a martingale. Similarly, using Eq.~(\ref{exx}), we know that 
\begin{eqnarray} \label{exx1}
d\exp(iuW_t^1+u^2t/2)=iu\exp(iuW_t^1+u^2t/2)dW_t^1,
\end{eqnarray}
which is, also, a martingale. Now, define the quantity
$$V_t:=\exp(iuW_t^1+u^2t/2+iuc_{-1} N_t^1-A(uc_{-1})t).$$ 
Given Eq.~(\ref{exx3}) and Eq.~(\ref{exx1}), It\^o's Lemma yields
\begin{equation}\label{exx4}
dV_t = \exp(iuW_t^1+u^2t/2+iuc_{-1} N_t^1-A(uc_{-1})t)\left(\int (\exp(iuc_{-1}\xi)-1)(\mu_1(dt,d\xi)-\lambda_1\nu_1(d\xi)dt)+ iudW_t^1\right).
\end{equation}
Therefore, if $j=1,$
\begin{eqnarray}\label{aa}
d\langle Jj,\Z_t\rangle V_t = V_t \langle J,\Z_t\rangle dW_t^{j}+\langle Jj,\Z_t\rangle dV_t +iuV_t\langle J,\Z_t\rangle dt,
\end{eqnarray}
and, if $j=0,$
\begin{eqnarray}\label{bb}
d\langle J0,\Z_t\rangle V_t=V_t\langle J,\Z_t\rangle dt +\langle J0,\Z_t\rangle dV_t.
\end{eqnarray}
Also, if $j=-1,$ we have
\begin{eqnarray}\label{cc}
d\langle J -1,\Z_t\rangle V_t&=&V_t\langle J,\Z_t\rangle dN^{1}_t +\langle J -1,\Z_t\rangle dV_t \notag\\
&+&V_t\langle J,\Z_t\rangle\left(\lambda_1\int \left(\exp(iuc_{-1}\xi)-1\right)\xi \nu_1(d\xi)\right)dt,
\end{eqnarray}
and, if $j=-q,$ for some $q \neq 1,$
\begin{eqnarray}\label{dd}
d\langle J -q,\Z_t\rangle V_t=V_t\langle J,\Z_t\rangle dN^{q}_t +\langle J -q,\Z_t\rangle dV_t.
\end{eqnarray}
Eqs. (\ref{aa}), (\ref{bb}), (\ref{cc}) and (\ref{dd}) imply that 
 \begin{align*}
(\langle Jj,\Z_t\rangle V_t)=
&\int_0^tV_s\langle J, \Z_s\rangle B_u(j)ds
+\text{martingale},
\end{align*}
 for $B_{u,c_{-1}}(j):=iu1_{\{j=1\}}+1_{\{j=0\}}+\lambda_1\int(\exp(iuc_{-1}\xi)\xi)\nu_1(d\xi)1_{\{j=-1\}}+\lambda_{-j}\int \xi\nu_{-j}(d\xi)1_{\{j<-1\}}$. This result, in particular, implies that
$$ \E_0[\langle Jj, \Z_t\rangle V_t]
=B_{u,c_{-1}}(j)\int_0^t \E_0\left[\langle J, \Z_s\rangle V_s 
\right]ds.$$
Because $\E_0[V_t]=1$, proceeding recursively over $I\in \{-e,\ldots,d\}^k,$ we obtain
\begin{align*}
\E_0[\langle I, \Z_t\rangle V_t]
&=
B_{u,c_{-1}}(i_1)\cdots B_{u,c_{-1}}(i_k)\int_0^t\int_0^{t_1}\cdots \int_0^{t_{k-1}}\E_0[V_{t_k}] d {t_k}\cdots d {t_1}\\
&=(iu)^{I(1)}\Big(\lambda_1\int(\exp(iuc_{-1}\xi)\xi)\nu_1(d\xi)\Big)^{I(-1)}\prod^e_{j=2}\left(\lambda_j\int \xi\nu_j(d\xi)\right)^{I(-j)}\frac {t^k}{k! },
\end{align*}
which proves the claim.
\end{proof}

\begin{proof}[Proof of Theorem~\ref{prop2J}]
Given an application of Taylor's theorem and Corollary~\ref{cor1new}, we have
$$\bigg|\E_0\Big[\exp\Big(iu\frac{X_t-c_0t}{c_1\sqrt t}\Big)\Big]
-\E_0\Big[\exp\Big(iu\frac{\langle  c,\Z_t\rangle-c_0t}{c_1\sqrt t}\Big)\Big]\bigg|\leq \frac {C}{\sqrt t}\E_0[|\e_3(t)|]=o(t).$$
We can, therefore, replace $X_t$ with its expansion $\langle  c,\Z_t\rangle$. Next, note that
$$\frac{\langle c,\Z_t\rangle -c_0t}{c_1\sqrt t }
=\frac{c_1W^1_t+ c_{-1}N^1_t+ c_{-2}N^2_t}{c_1\sqrt t}+\frac{\langle \overline c,\Z_t\rangle}{\sqrt t },$$
where $\overline c_1=\overline c_{-1}=\overline c_{-2}=\overline{c}_0=0$ ($c_0$ is also equal to zero - given Eq.~(\ref{eqnX}) - without loss of generality), which implies $\exp\Big(iu\frac{\langle c,\Z_t\rangle - c_0t}{\sqrt t}\Big) = \exp\Big(iu\frac{c_1W^1_t+ c_{-1}N^1_t+ c_{-2}N^2_t}{c_1\sqrt t}\Big)\exp\Big(iu\frac{\langle \overline c,\Z_t\rangle}{\sqrt t }\Big).$

Another application of Taylor's theorem allows us to expand the second term in the above expression:
$$\exp\Big(iu\frac{\langle \overline c,\Z_t\rangle}{\sqrt t}\Big)
=1+\frac{iu}{\sqrt t}\langle \overline c,\Z_t\rangle
+
\frac{(iu)^2}{ 2t}\langle \overline c, \Z_t\rangle^2+ R_t,
$$
where $|R_t|\leq \frac C {t^{3/2}}|\langle \overline c,\Z_t\rangle^3|$.
Since $\overline c_I\neq 0$ only if $2I(0)+I(>0)\geq 2$, by Corollary~\ref{cor1new} we obtain
$$\E_0[|\langle \overline c,\overline \Z_t\rangle|^3]\leq \E_0[|\langle \overline c,\overline \Z_t\rangle|^4]^{3/4}\leq (C t^{2\times 2})^{3/4}=O(t^3),$$
and, hence,
$$\E_0\left[\left|\exp\Big(iu\frac{c_1W_t^1+ c_{-1}N_t^1+c_{-2}N_t^2}{c_1 \sqrt t}\Big)\right|\left|R_t\right|\right]\leq \mathbb{E}_0[|R_t|] =o(t).$$
 The claim now follows from the definition of the operator $\overline \star$.
 \end{proof}

\begin{proof}[Proof of Corollary \ref{corr}]
Set $V_t= \exp(iu\frac{c_1W_t^1+ c_{-1}N_t^1+c_{-2}N_t^2}{c_1 \sqrt t}).$ Recall that 
$\overline c_1=\overline c_{-1}=\overline c_{-2}=c_0=0$ and $\overline c_I=c_I/c_1$ for each $I\notin\{-2,-1,0,1\}.$ Using Theorem~\ref{prop2J} write
    \begin{align}\label{ggg}
&\E_0\Big[\exp\Big(iu\frac{X_t-c_0t}{c_1\sqrt t}\Big)\Big] \\
&=
\E_0[V_t] \notag \\
&+\frac{iu}{ t^{1/2} }
\Big(\overline c_{11}\E_0[\langle (1,1),\Z_t\rangle V_t]
+\overline c_{10}\E_0[\langle (1,0),\Z_t\rangle V_t]
+\overline c_{01}\E_0[\langle (0,1),\Z_t\rangle V_t]
+\overline c_{111}\E_0[\langle (1,1,1),\Z_t\rangle V_t]\Big) \notag \\
&+ \frac{(iu)^2}{2 t }
\Big(\overline c_{11}^2\E_0[\langle (1,1)\overline \star(1,1),\overline\Z_t\rangle V_t] + \overline c_{21}^2\E_0[\langle (2,1)\overline \star(2,1),\overline\Z_t\rangle V_t]
\Big) \notag \\
 &+o(t) \notag.
 \end{align}
 We note that, in the second term on the right-hand side of Eq.~(\ref{ggg}), the terms with coefficients $\overline c_{0}$, $\overline c_{1}$, $\overline c_{-1}$ and $\overline c_{-2}$ do not appear since those coefficients are zero and the terms with coefficients $\overline c_{00}, \overline c_{000}, \overline c_{001}$ etc. are folded into the error term. Consider, e.g., the term associated with $\overline c_{001}.$ By Corollary \ref{corollario}, its associated Fourier transform is of order $t^{5/2}$. Once we standardize by $\frac{1}{t^{1/2}}$, the term is of order $t^2$ and, therefore, $o(t).$ 
 
A similar logic applies to the third term on the right-hand side of Eq.~(\ref{ggg}), which also derives from the properties of the $\overline \star$ operator. The first term, e.g., hinges on the fact that $11\overline \star11 = 6(1111) + 2(011)+2(101)+2(110)+00,$ given Lemma \ref{newprop}.

Setting $A(u):=\lambda_1\int (\exp(iu\frac{c_{-1}}{c_1}\xi)-1) \nu_1(d\xi)+\lambda_2\int (\exp(iu\frac{c_{-2}}{c_1}\xi)-1) \nu_2(d\xi)$ and recalling Lemma~\ref{lem2J2}, we now have
\begin{eqnarray*}
&&\E_0\Big[\exp\Big(iu\frac{X_t-c_0t}{c_1\sqrt t}\Big)\Big]\\
&&=\exp\Big(-\frac{u^2}2+A(u)t\Big)\left(1+\Big[- \frac{c_{11}}{c_1}\frac{iu^3}2\Big]{\sqrt t}\right.\\
&\qquad&\left.+\frac 1 2\Big [-\Big(\frac{ c_{01}}{c_1}+\frac{c_{10}}{c_1}\Big)u^2
+\Big(\frac{ c_{11}}{c_1}\Big)^2\Big(-\frac 1 2 u^2+u^4-\frac{1}{4}u^6\Big)
+\Big(\frac{ c_{21}}{c_1}\Big)^2\Big(\frac {1}{3}u^4-\frac {1}{2}u^2\Big)+\frac{ c_{111}}{c_1}\frac{u^4}{3} \Big]t \right)\\
&&+o(t).
\end{eqnarray*}
The claim follows from noticing that $\exp(A(u)t)=1+A(u)t+o(t)$.
\end{proof}

\end{small}

\pagebreak

\bibliographystyle{abbrvnat}

\end{document}